\documentclass[sigconf]{acmart}
\usepackage{amsmath}
\usepackage{bbm}
\usepackage{graphicx}
\usepackage{subfig}
\usepackage[ruled]{algorithm2e}
\usepackage{hyperref}
\usepackage{diagbox}
\usepackage{balance}

\DeclareMathOperator*{\argmax}{argmax}

\DeclareMathOperator{\mE}{\mathbb{E}}
\newtheorem{theorem}{Theorem}
\newtheorem{definition}{Definition}

\newcommand{\algnameNS}{TSFD Rank}
\newcommand{\algname}{\algnameNS\ }

\AtBeginDocument{%
  \providecommand\BibTeX{{%
    \normalfont B\kern-0.5em{\scshape i\kern-0.25em b}\kern-0.8em\TeX}}}

\setcopyright{acmcopyright}
\copyrightyear{2021}
\acmYear{2021}
\acmDOI{10.1145/3471158.3472260}

\acmConference[ICTIR '21]{Proceedings of the 2021 ACM SIGIR International Conference on the Theory of Information Retrieval}{July 11, 2021}{Virtual Event, Canada}
\acmBooktitle{Proceedings of the 2021 ACM SIGIR International Conference on the Theory of Information Retrieval (ICTIR '21), July 11, 2021, Virtual Event, Canada}
\acmPrice{15.00}
\acmDOI{10.1145/3471158.3472260}
\acmISBN{978-1-4503-8611-1/21/07}

\begin{document}
\fancyhead{}
%\title{Fairness and Diversity for Rankings in Two-Sided Markets}

\title{User Fairness, Item Fairness, and Diversity for Rankings in Two-Sided Markets}

\author{Lequn Wang}
\email{lw633@cornell.edu}
\affiliation{%
  \institution{Department of Computer Science, Cornell University}
    \city{Ithaca}
  \state{New York}
  \country{USA}
}
\author{Thorsten Joachims}
\email{tj@cs.cornell.edu}
\affiliation{%
  \institution{Department of Computer Science, Cornell University}
     \city{Ithaca}
  \state{New York}
  \country{USA}
}


\begin{abstract}
Ranking items by their probability of relevance has long been the goal of conventional ranking systems. While this maximizes traditional criteria of ranking performance, there is a growing understanding that it is an oversimplification in online platforms that serve not only a diverse user population, but also the producers of the items.
In particular, ranking algorithms are expected to be fair in how they serve all groups of users --- not just the majority group --- and they also need to be fair in how they divide exposure among the items. 
These fairness considerations can partially be met by adding diversity to the rankings, as done in several recent works. However, we show in this paper that user fairness, item fairness and diversity are fundamentally different concepts. In particular, we find that algorithms that consider only one of the three desiderata can fail to satisfy and even harm the other two.
To overcome this shortcoming, we present the first ranking algorithm that explicitly enforces all three desiderata. The algorithm optimizes user and item fairness as a convex optimization problem which can be solved optimally. From its solution, a ranking policy can be derived via a novel Birkhoff-von Neumann decomposition algorithm that optimizes diversity. 
Beyond the theoretical analysis, we investigate empirically on a new benchmark dataset how effectively the proposed ranking algorithm can control user fairness, item fairness and diversity, as well as the trade-offs between them. 
\end{abstract}



%   \includegraphics[width=\textwidth]{figures/movie_ranking_example.png}
%   \caption{Movie recommendation \lw{Have not found a better one yet.Companies do have diversity algorithms in itemion.}}
%   \Description{movie recommendation that presents producer fairness, consumer fairness, and diversity problems. }

\maketitle

\section{Introduction}
We consider ranking problems that involve two-sided markets of producers and consumers. Such two-sided markets are widespread in online platforms --- movie producers and audiences in a streaming platform, job seekers and employers in a resume database, or news agencies and information seekers in a news-feed app. In these two-sided markets, the items compete with each other for exposure to the users, while the users gain utility from the recommender system by finding items they like.
The platform mediates this market through the ranking algorithm, with great influence on which users get exposed to which items.

Conventional ranking algorithms maximize the average utility to the users by following the probability ranking principle (PRP)~\cite{robertson1977probability}. 
However, there is growing understanding that this is an oversimplification in online platforms that mediate a two-sided market. 
First, the objective of maximizing the average utility can unfairly marginalize minority user groups, decreasing how useful the ranking system is to them in order to better serve a majority user group \cite{xiao2017fairness}.
Second, the items compete with each other for exposure to the users, and there is the need to divide the exposure between the items in a fair way. It was shown that maximizing the average utility to the users can be unfair to the items, and that it can lead to winner-takes-all dynamics that amplify existing inequities~\cite{Singh/Joachims/18a}.
Violating user and/or item fairness is not only ethically fraught for many applications, it may also drive users and items from the platform \cite{evans2016matchmakers}, or violate anti-discrimination law~\cite{noble2018algorithms}, anti-trust law~\cite{MarkScott17}, or freedom of speech principles~\cite{grimmelmann2013speech}. 

Diversification of search results has often been employed to address these concerns, as well as related issues like super-star economics \cite{mehrotra2018towards}, perpetuation of stereotypes~\cite{beede2011women,kay2015unequal}, ideological polarization~\cite{bakshy2015exposure} and spread of misinformation~\cite{vosoughi2018spread,tabibiandesign}. However, while diversification appears related to fairness at first glance, it is not clear whether standard formalizations of diversity \cite{radlinski2008learning,kulesza2012determinantal} actually achieve fairness and vice versa.

In this paper we provide the first theoretical study of the interplay between user fairness, item fairness and diversity for rankings in two-sided markets. To enable this theoretical analysis, we quantify the three desiderata in an intent-aware setup~\cite{clarke2008novelty,agrawal2009diversifying,chapelle2011intent} where users have different intents and items have varying relevance to different intents. In particular, we formalize user fairness as an economic social-welfare objective where user groups differ in their intent distributions, and relate this to submodular diversity objectives. For item fairness, we adapt the disparate treatment constraints proposed in ~\cite{Singh/Joachims/18a} for the intent-aware setup to ensure the exposure is fairly allocated to the item groups based on their merit. 

Through this theoretical analysis, we show that user fairness, item fairness and diversity are independent goals. Specifically, algorithms that optimize any one of the desiderata can fail to satisfy and even harm the other two. 

To address this problem, we present the first ranking algorithm that explicitly enforces all three desiderata --- called \algname for Two-Sided Fairness and Diversity. The algorithm optimizes user and item fairness as a convex optimization problem which can be solved optimally. From its solution, a ranking policy can be derived via a novel Birkhoff-von Neumann decomposition~\cite{birkhoff1940lattice} algorithm that optimizes diversity.

In addition to the theoretical analysis, we constructed the first benchmark dataset with annotations for intents, user groups and item groups. On this dataset, we empirically evaluate the proposed \algname with ablation studies quantifying the dependencies between user fairness, item fairness and diversity.

\section{Related Work}

As algorithmic techniques, especially machine learning, are increasingly used to make decisions that directly impact people's life, there is growing interest in understanding their societal impact. Many works proposed mathematical desiderata to test algorithmic fairness in binary classification~\cite{HardtPNS16,chouldechova2017fair,KleinbergMR17,agarwal2018reductions}. These desiderata often operationalize definitions of fairness from political philosophy and sociology. 
We study the societal impact of the less explored problem of ranking which, unlike binary classification, is a structured output prediction problem with an exponentially large output space. 
Since users have different preferences and items compete for exposure in the rankings, the fairness definitions from binary classification do not directly translate to ranking problems. 

Unfairness in rankings typically comes from two sources. Some works focus on the endogenous design of the fair ranking systems~\cite{yang2017measuring,Singh/Joachims/18a}. They answer what a fair ranking system is and how to achieve fairness assuming all the system information such as the relevance and the position bias is known. The second source of unfairness are the exogenous factors such as biases in the data~\cite{bottou2013counterfactual,Joachims/etal/17a,yang2018unbiased,rastegarpanah2019fighting} and biases during relevance estimation~\cite{yao2017beyond,burke2018balanced}. 
Some works take both into consideration~\cite{yadav2019fair,zehlike2020reducing}. 
We focus on the endogenous design of fair and diverse rankings in two-sided markets, which is orthogonal to exogenous factors.  

Most existing works on fairness in rankings consider item fairness. They can be classified into three types: (1) composition-based item fairness which ensures statistical parity of where the items are ranked~\cite{yang2017measuring, zehlike2017fa, celis2017ranking,stoyanovich2018online,asudeh2019designing, geyik2019fairness, celis2020interventions}; (2) pairwise-comparison-based item fairness which aims for statistical parity of pairwise ranking errors between item groups~\cite{beutel2019fairness,kallus2019fairness,narasimhan2020pairwise}; and (3) merit-exposure-based item fairness which explicitly quantifies the amount of exposure an item gets in a ranking and allocates exposure to the items based on their merit~\cite{Singh/Joachims/18a, biega2018equity, mehrotra2018towards,Singh/Joachims/19a,sapiezynski2019quantifying, morik2020controlling}. We adopt the third type of item fairness since (1) unlike composition-based item fairness, it can allocate exposure based on merit; (2) unlike pairwise-comparison-based item fairness, it takes position bias into consideration; and (3) it explicitly quantifies the amount of exposure an item gets in a ranking which enables the quantifiable study of the relationship of item fairness to user fairness and diversity. 

Fewer works consider user fairness in rankings and they often consider user unfairness problems that result from achieving item fairness across different queries. Some works~\cite{patro2020fairrec,basu2020framework} propose to fairly share the utility drop among the user groups when achieving item fairness across queries. Patro et al.~\cite{patro2020incremental} regard the drastic change of exposure to the items during the policy updates to be unfair, and propose an online update algorithm to smoothly update the policy so that the exposure to the items changes gradually while ensuring a minimum utility for the users during the policy updates. In contrast, we identify user unfairness problems originating from the user intent difference and uncertainty for an individual query, which exists even when we do not consider item fairness. Some works consider user fairness in group recommendation, where a recommendation needs to satisfy a group of users with different preferences~\cite{xiao2017fairness}. They assume the relevance of each item to each user is known. We model the user preferences and the associated uncertainty in an intent-aware setup. 

Diversity in rankings and recommendations also challenges the PRP. The key mechanism behind diversity is to model the utility as a function that is not modular (i.e. linearly additive) in the set of ranked items, but that exhibits a diminishing-returns property --- most commonly in the form of a submodular set function~\cite{radlinski2008learning,chapelle2011intent,yue2011linear}. In extrinsic diversity ~\cite{carbonell1998use,radlinski2008learning,zhai2015beyond}, this is used to hedge against the uncertainty
about the user's information need; and in intrinsic diversity ~\cite{clarke2008novelty,radlinski2009redundancy} this is used to model
complementarity and substitution in a sense of portfolio optimization. Since we are dealing with uncertain user intents, our goal is to achieve extrinsic diversity for the rankings. 

Two-sided platforms are modeled as matchmakers that reduce the friction between the two sides of the market. The key to the success of two-sided platforms is to ensure a critical mass of participants on both sides, since they are in need of each other. Literature in economics~\cite{caillaud2003chicken,rochet2003platform,armstrong2006competition,evans2016matchmakers} focuses on the effect of business strategies, primarily about pricing, on the two-sided markets, but typically does not model the effect that the platform's ranking algorithm has on the interactions between users and items. Recently, some works~\cite{burke2017multisided,stanton2019revenue, abdollahpouri2020multistakeholder} advocate viewing recommendation problem in the context of two-sided markets and discussed fairness issues on both sides. But neither mathematical definitions nor theoretical characterizations of fairness are provided.  

The algorithmic study of two-sided matching markets dates back to the stable matching algorithm analyzed by Gale and Shapley~\cite{gale1962college}. Some works propose algorithms in this context to select a fair stable matching from a set of stable matchings~\cite{masarani1989existence,klaus2006procedurally}. Recently, S{\"u}hr et al.~\cite{suhr2019two} consider fairness concerns in ride-hailing platforms and propose an online matching algorithm to ensure salary fairness for the drivers amortized over time. In reciprocal recommendation problems, the success is measured by the satisfaction on both sides of the market such as in online dating platforms~\cite{DBLP:conf/recsys/PizzatoRCKK10}. We consider problems where the items have no preferences over the users and where there are no supply constraints on the items. 

One key aspect of fair rankings is the fair division of the exposure to the users among the items~\cite{Singh/Joachims/18a,biega2018equity,Wang/etal/21a}. Fair division~\cite{brams1996fair,steihaus1948problem,procaccia2013cake} has been studied for decades where the goal is to allocate a set of resources to the agents. Two
of the classic desiderata for fair division are (1) proportionality i.e. every agent receives its ``fair share'' of the utility, and (2) envy-freeness i.e. no agent wishes to swap her allocation with another agent. In the proposed ~\algnameNS, the optimization of user and item fairness can be thought of as ensuring proportionality for the users and the items, and the optimization for diversity can be interpreted as reducing the envy of the users. 

\section{Ranking in a Two-Sided Market} \label{sec:rankingsetup}

As a basis for our theoretical analysis, as well as the derivation of the \algnameNS\ algorithm, we first formalize the problem of intent-aware ranking in two-sided markets. This includes definitions of utility, user fairness, item fairness and diversity. 

\subsection{Intent-Aware Setup and Utility}

We consider the problem of ranking a set of items $\mathbbm{D}^q = \{d_1, d_2,$ $ d_3, ... \}$ to present to a user with query $q$. A query can be a text query (e.g. ``Schwarzenegger'')  or any other context for ranking (e.g. ``featured movies today''). 
In the intent-aware setup \cite{clarke2008novelty,agrawal2009diversifying,chapelle2011intent}, each user has an unobserved intent $i\in\mathbbm{I}$ that further refines the query (e.g. preferred movie genre). We denote with $\mathcal{I}_{\mathcal{U}^q}$ the intent distribution of the user population $\mathcal{U}^q$ for query $q$. Each item has varying relevance to different intents (e.g. a movie has varying relevance to different genres), and we denote the relevance of an item $d$ to an intent $i$ as $r(d, i)$. A wide range of existing methods can be used to learn this relevance function, and we merely assume that relevance estimates $r(d, i)$ are available. A ranking $\sigma$ is a permutation of the set of items and a ranking policy $\pi(\cdot|q)$ is a probability distribution over all possible permutations, where deterministic ranking policies are a special case. 
Focusing on additive ranking metrics, the utility of a ranking policy $\pi$ is 
\begin{equation}
    U(\pi|q) = \mE_{i\sim\mathcal{I}_{\mathcal{U}^q},\sigma\sim\pi(\cdot|q)}\left[\sum_{d\in\mathbbm{D}^q}e(\sigma(d))r(d,i)\right], \label{eq:util_overall}
\end{equation}
where $\sigma(d)$ is the rank of the item $d$ and $e$ maps this rank to the exposure $d$ will receive in the position-based model \cite{craswell2008experimental}. 
Since a user has limited attention for each ranking, we assume the total exposure is bounded i.e. $0<\sum_{m}e(m)<\infty$. 
While \eqref{eq:util_overall} involves an expectation in the exponential space of rankings, $U(\pi|q)$ can equivalently be written in terms of a marginal rank probability matrix $\Sigma^{\pi,q}$ with
\begin{equation}
    \mathbf{\Sigma}^{\pi,q}_{m,n} = \mE_{\sigma\sim\pi(\cdot|q)}\left[\mathbbm{1}_{\{\sigma(d_m)=n\}}\right]\quad\forall m,n
\end{equation}
where each entry represents the probability of ranking item $d_m$ at rank $n$ under policy $\pi$ 
\begin{equation}
\begin{split}
    U(\pi|q) &=
    \left(\mathbf{r}^{\mathcal{U}^q}\right)^{\top}\Sigma^{\pi,q}\mathbf{e}.
\end{split}
\end{equation}
$\mathbf{r}^{\mathcal{U}^q}$ is the vector containing the expected relevance of each item to the whole user population with $\mathbf{r}_m^{\mathcal{U}^q} = \mE_{i\sim\mathcal{I}_{\mathcal{U}^q}}[r(d_m,i)]$, and $\mathbf{e}$ is the exposure vector with $\mathbf{e}_n=e(n)$. The marginal rank probability matrix is doubly stochastic~\cite{birkhoff1940lattice} since the sum of each row and each column is $1$, i.e. $\sum_{m}\Sigma^{\pi,q}_{m,n}=1$ for all $n$ and $\sum_{n}\Sigma^{\pi,q}_{m,n}=1$ for all $m$. 

\subsection{User Fairness}
\label{sec:uf}
Overall utility as defined in Eq.~\eqref{eq:util_overall} reflects an average over all users. However, different user groups $UG\in\mathbbm{UG}$ (e.g. male vs. female users) can have different intent distributions $\mathcal{I}_{\mathcal{UG}^q}$ for a query $q$, and suboptimal ranking performance for a minority group may get drowned out. Since group membership is typically not known for privacy reasons~\cite{holstein2019improving}, a ranking policy needs to make sure that it does not unfairly provide disparate levels of utility to the user groups. 
We define the utility of a ranking policy $\pi$ for a user group $UG$ as the expected utility for the users in this group
\begin{equation}
\begin{split}
    U(\pi|UG,q) &= \mE_{i\sim\mathcal{I}_{\mathcal{UG}^q},\sigma\sim\pi(\cdot|q)}\bigg[\sum_{d\in\mathbbm{D}^q}e(\sigma(d))r(d, i)\bigg]\\
        &=\left(\mathbf{r}^{\mathcal{UG}^q}\right)^{\top}\Sigma^{\pi,q}\mathbf{e}, \label{eq:util_usergroup}
\end{split}
\end{equation}
where $\mathbf{r}^{\mathcal{UG}^q}$ is the vector containing the expected relevance of each item to the user group $UG$ with  $\mathbf{r}^{\mathcal{UG}^q}_m = \mE_{i\sim\mathcal{I}_{\mathcal{UG}^q}}[r(d_m,i)]$.
A fair ranking ensures that each user group gets an equitable amount of utility. In economics, the goal of providing an equitable amount of utility across groups is typically formalized through a social-welfare function \cite{vondrak2008optimal}, which is maximized to optimize fairness. We adopt
\begin{equation}
\begin{split}
    UF(\pi|q) &= \sum_{UG\in\mathbbm{UG}}\rho_{UG}^qf(U(\pi|UG,q))
\end{split}
\end{equation}
as our class of social-welfare functions, where $f$ is an increasing concave function (e.g. $\log$) that models the diminishing return property. This social-welfare objective provides larger return for increasing the utility of a user group with little utility compared to increasing the utility of a user group that already receives a large utility. Thus it encourages the ranking policy to provide more equal utility to each user group. $f$ can be chosen based on application requirements. $\rho_{UG}^q$ denotes the group proportion of $UG$, i.e. the probability that a user sampled from the whole user distribution $\mathcal{U}^q$ belongs to user group $UG$. Since $UF(\pi|q)$ is a convex combination of concave functions of $\Sigma^{\pi,q}$, user fairness is a concave function of $\Sigma^{\pi,q}$. 

Since the intent distribution of the overall user population is a convex combination of the intent distribution of each user group $\mathcal{I}_{\mathcal{U}^q} = \sum_{UG\in\mathbbm{UG}}\rho_{UG}^q\mathcal{I}_{\mathcal{UG}^q}$, 
user fairness $UF(\pi|q)$ is a lower bound of the overall utility in \eqref{eq:util_overall} after transformation through the inverse of the user fairness function $f^{-1}$ 
%which is an increasing convex function. 
\begin{displaymath}
\begin{split}
    &f^{-1}(UF(\pi|q))\\
    =& f^{-1}\left(\sum_{UG\in\mathbbm{UG}}\rho_{UG}^qf(U(\pi|UG,q))\right)\\
    \leq& f^{-1}\left(f\left(\sum_{UG\in\mathbbm{UG}}\rho_{UG}^qU(\pi|UG,q)\right)\right)\\
%   =&\sum_{UG\in\mathbbm{UG}}\rho_{UG}^qU(\pi|UG,q))\\
%   =&\sum_{UG\in\mathbbm{UG}}\rho_{UG}^q\mE_{i\sim\mathcal{I}_{\mathcal{UG}^q},\sigma\sim\pi(\cdot|q)}\left[\sum_{d\in\mathbbm{D}^q}e(\sigma(d))r(d, i)\right]\\
%   =&\mE_{i\sim\mathcal{I}_{\mathcal{U}^q},\sigma\sim\pi(\cdot|q)}\left[\sum_{d\in\mathbbm{D}^q}e(\sigma(d))r(d, i)\right]\\
   =&U(\pi|q).
\end{split}
\end{displaymath}

The inequality holds because $f$ is concave. This shows that maximizing user fairness is maximizing a lower bound of the overall utility from Eq.~\eqref{eq:util_overall}.

\subsection{Item Fairness}

Fairness to the items is akin to a fair-division problem. Specifically, items compete for exposure to the users, since exposure is a prerequisite for items to derive utility (e.g. revenue) from the ranking.
We adopt the disparate treatment constraints proposed in~\cite{Singh/Joachims/18a} for our theoretical and empirical analysis.  
The disparate treatment constraints ensure that each item group $DG$ gets an amount of exposure $E(\pi|DG,q)$ that is proportional to its merit $M(DG,q)>0$ 
\begin{equation}
\frac{E(\pi|DG_m,q)}{M(DG_m,q)} = \frac{E(\pi|DG_n,q)}{M(DG_n,q)}\quad\forall m,n.  \label{eq:item_fairness_naive}
\end{equation}

Further specifying $E(\pi|DG,q)$ and $M(DG,q)$, the average exposure of an item group is defined as
\begin{equation}
\begin{split}
\label{equ:exposure}
    E(\pi|DG,q)&= \mE_{ \sigma\sim\pi(\cdot|q)}\left[\frac{1}{|DG|}\sum_{d\in DG}e(\sigma(d))\right]\\
    &=\frac{\left(\mathbf{l}^{DG}\right)^{\top}\Sigma^{\pi,q}\mathbf{e}}{|DG|}, 
\end{split}
\end{equation}
where $\mathbf{l}^{DG}$ is the label vector that denotes whether an item belongs to item group $DG$ with $\mathbf{l}^{DG}_m = \mathbbm{1}_{\{d_m\in DG\}}$. 
For the empirical evaluation, we adopt the average relevance of the items in the item group as the merit function
\begin{equation}
    M(DG,q) = \mE_{i\sim\mathcal{I}_{\mathcal{U}^q}}\bigg[\frac{1}{|DG|}\sum_{d\in DG}r(d, i)\bigg].
\end{equation}
In practice, the merit function can be chosen based on application-specific requirements. 

Finally, to quantify that items also draw utility from the rankings, we define the utility of an item group $DG$ as 
\begin{equation}
\begin{split}
    U(\pi|DG,q) &= \mE_{i\sim\mathcal{I}_{\mathcal{U}^q}, \sigma\sim\pi(\cdot|q)}\bigg[\sum_{d\in DG}e(\sigma(d))r(d,i)\bigg]\\
    &= \left(\mathbf{l}^{DG}\circ \mathbf{r}^{\mathcal{U}^q}\right)^{\top}\Sigma^{\pi,q}\mathbf{e}, \label{eq:util_itemgroup}
\end{split}
\end{equation}
where $\circ$ denotes the element-wise product.
In the position-based click model~\cite{craswell2008experimental}, $U(\pi|DG,q)$ corresponds to the sum of the click-through rates of the items in item group $DG$ under ranking policy $\pi$. 

\subsection{Diversity}
\label{sec:diversity}

The original and dominant motivation for diversity in ranking arises from the uncertainty about the user's intent \cite{carbonell1998use}. To hedge against this uncertainty, a diversified ranking aims to provide utility no matter what the unknown intent of the user is (i.e. extrinsic diversity \cite{radlinski2009redundancy}). To formalize this goal, we first define the utility of a ranking $\sigma$ for an intent $i$ with an additive metric analogous to the overall utility in Eq.~\eqref{eq:util_overall} as
\begin{equation}
    U(\sigma|i,q) = \sum_{d\in\mathbbm{D}^q}e(\sigma(d))r(d,i).
\end{equation}
Similar to user fairness, the diversity $D(\sigma|q)$ of a ranking $\sigma$ is typically quantified using an increasing concave function $g$ that encourages each ranking in the ranking policy to cover multiple intents \cite{radlinski2008learning,agrawal2009diversifying,chapelle2011intent,yue2011linear}
\begin{equation}
    D(\sigma|q) = \mE_{i\sim\mathcal{I}_{\mathcal{U}^q}}\left[g(U(\sigma|i,q))\right].
\end{equation}
Consequently, for a ranking policy $\pi$, the expected diversity is
\begin{equation}
    D(\pi|q) = \mE_{\sigma\sim\pi(\cdot|q), i\sim\mathcal{I}_{\mathcal{U}^q}}\left[g(U(\sigma|i,q))\right]. 
\end{equation}

Diversity and user fairness differ in two fundamental ways. First, user fairness aggregates over user groups, while diversity aggregates over intents. Second, user fairness amortizes over intents and draws from $\pi$ as input to the concave function, while diversity takes the expectation over intents after the concave transformation. This adds emphasis on optimizing each individual ranking in the diversity objective.
It also implies that diversity $D(\pi|q)$ can not be written as a linear function of $\Sigma^{\pi,q}$. Furthermore, unlike utility and user fairness, two ranking policies $\pi$ and $\pi'$ that both produce the same marginal rank probability matrix $\Sigma^{\pi,q}=\Sigma^{\pi',q}$ can have different diversity $D(\pi|q)\not=D(\pi'|q)$.

Similar to user fairness, diversity is also a lower bound on the overall utility from Eq.~\eqref{eq:util_overall} after transformation with the inverse function $g^{-1}$
\begin{displaymath}
\begin{split}
g^{-1}(D(\pi|q))&=g^{-1}\left(\mE_{i\sim\mathcal{I}_{\mathcal{U}^q},\sigma\sim\pi(\cdot|q)}[g(U(\sigma|i,q))]\right)\\
&\leq g^{-1}\left(g\left(\mE_{i\sim\mathcal{I}_{\mathcal{U}^q},\sigma\sim\pi(\cdot|q)}[U(\sigma|i,q)]\right)\right)\\
% &= g^{-1}(g(\mE_{i\sim\mathcal{I}_{\mathcal{U}^q},\sigma\sim\pi(\cdot|q)}\left[\sum_{d\in\mathbbm{D}^q}e(\sigma(d))r(d,i)\right]))\\
&=U(\pi|q), 
\end{split}
\end{displaymath}
where the inequality holds because $g$ is concave. This indicates that maximizing diversity also maximizes a lower bound on the overall utility. Diversity maximization can be expressed as a submodular maximization problem with two matroid constraints ~\cite{krause2014submodular}, and we will optimize it using the standard greedy approximation algorithm in our experiments. For completeness, the algorithm is detailed in the appendix.

\section{Theoretical Analysis}
In this section, we analyze the interplay between utility, user fairness, item fairness and diversity. First, we provide worse-case analyses showing that individual user groups, item groups, or intents can needlessly receive zero utility if their interests are not explicitly represented in the ranking objective. This indicates that user fairness, item fairness and diversity are fundamentally different objectives and achieving one of them does not automatically satisfy another. 
Second, we develop a new form of utility-efficiency analysis to show that achieving one of user fairness, item fairness and diversity might fail to satisfy the utility efficiency of the others. This suggests that the utility efficiencies of the three desiderata are in conflict with each other and achieving one of them might harm the other two. 

\subsection{Zero-Utility Analysis}

Our zero-utility analysis investigates whether a user group, item group or intent can receive a utility of zero, even if a non-zero solution exists. 
For clarity, we first define a class of non-degenerate ranking problems, to focus our theoretical analysis on non-degenerate cases where non-zero solutions exist. 
\begin{definition}{(Non-degenerate ranking problem)}
A ranking problem represented by a tuple $(\mathbbm{I},$ $\mathbbm{UG},$ $\mathbbm{DG},$ $\mathbbm{D}^q$ $\mathcal{U}^q,$ $r,$ $e)$ is non-degenerate if (1) every user group $UG$ has positive group proportion $\rho_{UG}^q>0 $; (2) every intent $i$ has positive probability mass in the user intent distribution $\mathcal{I}_{\mathcal{U}^q}(i)>0$; (3) for every intent $i$, there exists an item $d$ that has positive relevance for the intent $r(d,i)>0$; and (4) for every item group $DG$, there exists an item $d\in DG$ such that the expected relevance of the item is positive $\mathbbm{E}_{i\sim\mathcal{I}_{\mathcal{U}^q}}[r(d, i)]> 0$.
\end{definition}
For user groups and item groups, we investigate whether every group achieves non-zero utility as defined in Eqs.~\eqref{eq:util_usergroup} and ~\eqref{eq:util_itemgroup} under different policies. For the intents, since diversity is a function of individual rankings and each single ranking might not be able to provide non-zero utility for every intent due to limited number of non-zero exposure positions, we define the amount of intent covered by a ranking $\sigma$ as the amount of intent that has non-zero utility
\begin{displaymath}
\sum_{i\in \{i | i\in\mathbbm{I}, U(\sigma|i,q)>0\}}\left[ \mathcal{I}_{\mathcal{U}^q}(i)\right], 
\end{displaymath}
and investigate whether each ranking sampled from a ranking policy covers the maximum amount of intent covered by any ranking 
\begin{displaymath}
max_\sigma \sum_{i\in \{i | i\in\mathbbm{I}, U(\sigma|i,q)>0\}}\left[ \mathcal{I}_{\mathcal{U}^q}(i)\right]. 
\end{displaymath}

We present two example theorems for this zero-utility analysis. Theorem~\ref{theo:util_zero_user} shows that there exist ranking problems where maximizing overall utility needlessly provides zero utility for some user groups. 

\begin{theorem}
\label{theo:util_zero_user}
There exist non-degenerate ranking problems such that any ranking policy $\pi$ maximizing overall utility $U(\pi|q)$ has utility $U(\pi|UG,q)=0$ for some user group $UG$. 
\end{theorem}

Proofs of all theorems are provided in the appendix. 
The disparate treatment identified in Theorem ~\ref{theo:util_zero_user} is not necessary, since the following Theorem~\ref{theo:uf_zero_user} shows that directly maximizing user fairness can always ensure non-zero utility for every user group. 

\begin{table}[t]
\caption{Summary of zero-utility analysis.}
\label{tab:theory_zero_util}
\scalebox{1.0}{
\begin{tabular}{l|lll}
\begin{tabular}[c]{@{}l@{}}The policy\\ optimizing \end{tabular} 
& \begin{tabular}[c]{@{}l@{}}Non-zero utility\\ for every\\user group?\end{tabular} 
& \begin{tabular}[c]{@{}l@{}}Non-zero utility\\ for every \\item group?\end{tabular} 
& \begin{tabular}[c]{@{}l@{}}Rankings cover\\ maximum  \\  intent?\end{tabular}  \\
\hline
Utility & $\times$ & $\times$ & $\times$ \\
User fairness & \checkmark & $\times$ & $\times$\\
Item fairness & $\times$ & \checkmark & $\times$\\
Diversity   & $\times$ & $\times$ & \checkmark\\
\algname & \checkmark& \checkmark& $\times$
\end{tabular}}
\end{table}

\begin{theorem}
\label{theo:uf_zero_user}
For any non-degenerate ranking problem, there exists a user fairness function $f$ such that if a ranking policy $\pi$ maximizes user fairness $UF(\pi|q)$, then every user group has non-zero utility under this ranking policy $\pi$. 
\end{theorem}

\autoref{tab:theory_zero_util} summarizes the other formal results, which are detailed in the appendix. For the sake of brevity, we use ``maximize item fairness'' to refer to the more appropriately descriptive ``maximize utility subject to the disparate treatment constraints''.
We provide an intuition of the analysis through the ranking problem in \autoref{fig:example}.  

\textbf{First, maximizing overall utility can lead to zero utility for a user group and/or for an item group, and it can fail to cover the maximum amount of intent.} Since items in $d_{3*}$ have strictly larger expected relevance to the whole user population than all the other items,
a ranking policy that maximizes utility will rank items in $d_{3*}$ over all the other items. If only $3$ positions have non-zero exposure, items in $d_{3*}$ will occupy all the $3$ non-zero exposure positions and thus leave zero exposure for the other items. This leads to a ranking policy with zero utility for $UG_1$, $DG_2$, $i_1$ and $i_2$. We can easily construct a ranking that covers all the intents by selecting three items that are relevant to the three intents respectively and putting them in the $3$ non-zero exposure positions. Thus any ranking sampled from the policy maximizing utility fails to cover the maximum amount of intent, since the ranking will only cover intent $i_3$. 

\begin{figure}[!tbp]
  \centering
  \includegraphics[width=0.49\textwidth]{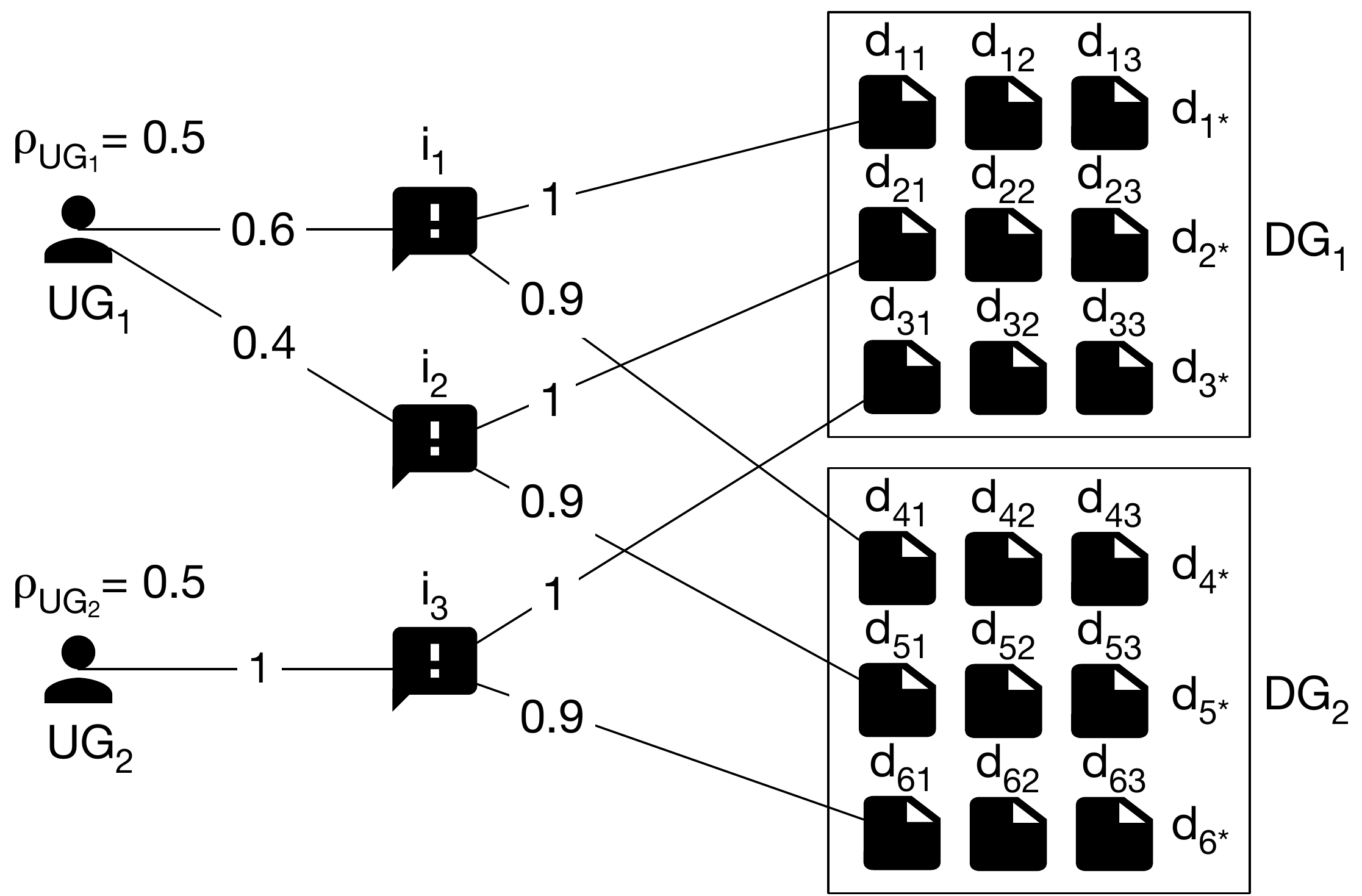}
  \caption{A ranking problem that illustrates the zero-utility analysis. The ranking problem consists of $2$ user groups, $3$ intents and $2$ item groups. $6$ sets of items are partitioned into the $2$ item groups as indicated by the squares. Each set consists of $3$ items with exactly the same relevance to every intent and we denote with $d_{m*}$ the $m^{th}$ set. The numbers on the edges between user groups and intents represent the intent distribution. The numbers on the edges between intents and items represent relevance. For clarity, we omit edges with $0$ probability or $0$ relevance.}
   \label{fig:example}
\end{figure}

\textbf{Second, enforcing item fairness can lead to zero utility for a user group and fail to cover the maximum amount of intent.}
Similarly, maximizing item fairness would rank items in $d_{3*}$ over items in $d_{1*}$, $d_{2*}$ and rank items in $d_{6*}$ over items in $d_{4*}$, $d_{5*}$, 
since items in $d_{3*}$ and $d_{6*}$ have the largest expected relevance to the whole user population within each item group. If only $3$ positions have non-zero exposure, items in $d_{3*}$ and $d_{6*}$ will occupy the $3$ non-zero exposure positions. 
This leads to a ranking policy with zero utility for $UG_1$, $i_{1}$ and $i_{2}$. As discussed in the last paragraph, the maximum amount of intent that can be covered by a ranking is $1$. Thus maximizing item fairness fails to cover the maximum amount of intent since the rankings sampled from the policy maximizing item fairness only cover $i_3$. 

\textbf{Third, maximizing user fairness can lead to zero utility for an item group and fail to cover the maximum amount of intent.} Again, any ranking policy that maximizes user fairness would rank items in $d_{1*}$ over items in $d_{2*}$, $d_{4*}$, $d_{5*}$ and rank items in $d_{3*}$ over items in $d_{6*}$, 
since items in $d_{1*}$ and $d_{3*}$ have the largest expected relevance for the two user groups respectively, and all the items have positive relevance to only one user group. If only $3$ positions have non-zero exposure, then items in $d_{1*}$ and $d_{3*}$ will occupy all the $3$ non-zero exposure positions. This leads to a ranking policy with zero utility for $DG_2$ and $i_2$. Thus the rankings sampled from the policy maximizing user fairness fail to cover the maximum amount of intent. 

\textbf{Fourth, maximizing diversity can lead to zero utility for a user group and/or for an item group.}
If only $1$ position has non-zero exposure, maximizing diversity will always put one item from $d_{3*}$ in that non-exposure position since items in $d_{3*}$ have the largest relevance to $i_3$, the intent with the largest density. This leads to zero utility for $DG_2$ and $UG_1$.  

This worst-case analysis indicates that it is necessary to optimize each of user fairness, item fairness, and diversity, since any one criterion does not even provide the guarantee of non-zero utility for the others. 

\subsection{Utility-Efficiency Analysis}
\label{sec:pareto}
Our utility-efficiency analysis investigates if optimizing for one of user fairness, item fairness, or diversity can provide a utility-efficient solution for any of the other desiderata. 
To answer this kind of questions, we first introduce the precise meaning of utility efficiency for user groups, item groups and intents. 

For the user groups, we focus on utility Pareto efficiency of ranking policies. 
We begin by defining a dominance relation between two policies with respect to a multi-objective optimization problem. The objectives are the utilities of a ranking policy for the user groups $U(\pi|UG,q)$ from Eq.~\eqref{eq:util_usergroup}. 

\begin{definition}{(Dominance of ranking policies for the user groups)}
For a non-degenerate ranking problem,
a ranking policy $\pi$ dominates another ranking policy $\pi'$ for the user groups $\mathbbm{UG}$ if $U(\pi|UG,q)\geq U(\pi'|UG,q)$ for all $UG\in\mathbbm{UG}$ and there exists $UG\in\mathbbm{UG}$ such that $U(\pi|UG,q)>U(\pi'|UG,q)$. 
\end{definition}

The Pareto efficiency of a ranking policy for the user groups is then defined as follows.

\begin{definition}{(Pareto efficiency of ranking policies for the user groups)}
For a non degenerate ranking problem,
a ranking policy $\pi$ is Pareto efficient for the user groups $\mathbbm{UG}$ if $\pi$ is not dominated by any ranking policy $\pi'$  for $\mathbbm{UG}$. 
\end{definition}

For the intents, since diversity emphasises the performance of each ranking, we analyze the utility Pareto efficiency of rankings for the intents. 

\begin{definition}{(Dominance of rankings for the intents)}
For a non-degenerate ranking problem,
a ranking $\sigma$ dominates another ranking $\sigma'$ for the intents $\mathbbm{I}$ if $U(\sigma|i,q)\geq U(\sigma'|i,q)$ for all $i\in\mathbbm{I}$ and there exists $i\in\mathbbm{I}$ such that $U(\sigma|i,q)>U(\sigma'|i,q)$. 
\end{definition}

\begin{definition}{(Pareto efficiency of rankings for the intents)}
For a non-degenerate ranking problem,
a ranking $\sigma$ is Pareto efficient for the intents $\mathbbm{I}$ if $\sigma$ is not dominated by any ranking $\sigma'$  for $\mathbbm{I}$. 
\end{definition}

For the item groups, utility efficiency is achieved when items are ranked by their expected relevance to the whole user population within each item group, since otherwise we can switch the two items that are not ranked by their expected relevance to get larger utility for the item group they belong to without changing the exposure allocation among the item groups. 

\begin{definition}{(Items ranked by expected relevance within each item group)}
For a non-degenerate ranking problem and a ranking policy $\pi$, the items are ranked by their expected relevance to the whole user population within each item group under $\pi$ when for any $\sigma$ with $\pi(\sigma|q)>0$ and for all $DG\in\mathbbm{DG},d_m, d_n\in DG$, if $\mE_{i\sim\mathcal{I}_{\mathcal{U}^q}}\left[r(d_m,i)\right]>\mE_{i\sim\mathcal{I}_{\mathcal{U}^q}}\left[r(d_n,i)\right]$, then $e(\sigma(d_m))\geq e(\sigma(d_n))$.
\end{definition}

Achieving utility efficiency can be interpreted as not picking a solution that could easily be improved upon. Thus, if a procedure fails the test of utility efficiency, it clearly provides a suboptimal solution to the user groups, item groups or the intents. 
We present two example theorems that characterize the utility efficiency of optimizing overall utility, user fairness, item fairness, and diversity on the utility of users, items, and intents. Theorem~\ref{theo:uf_efficiency_user} shows that maximizing user fairness is Pareto efficient for the user groups. 

\begin{theorem}
\label{theo:uf_efficiency_user}
For any non-degenerate ranking problem and user fairness function $f$, if a ranking policy $\pi$ maximizes user fairness $UF(\pi|q)$, then $\pi$ is Pareto efficient for the user groups. 
\end{theorem}

While the solution is utility-efficient for the user groups, the following Theorem~\ref{theo:uf_efficiency_item} shows that this solution is not utility-efficient for the item groups.

\begin{theorem}
\label{theo:uf_efficiency_item}
There exists a ranking problem and a user fairness function $f$ such that items are not ranked by the expected relevance within each item group under any ranking policy $\pi$ that maximizes user fairness $UF(\pi|q)$. 
\end{theorem}

We summarize the results of our full utility efficiency analysis in \autoref{tab:theory_pareto} and provide the details in the appendix.  Maximizing overall utility is the only criterion that ensures utility efficiency for all groups, but the solutions may be poor in terms of fairness or diversity as shown in the zero-utility analysis. Once we explicitly optimize for any one of the fairness or diversity goals, the ranking policy is generally not utility-efficient for the other goals (except that maximizing user fairness ensures utility efficiency for the intents). This implies that the utility efficiency of the three goals are in conflict with each other. Optimizing one of the three desiderata might cause harm or utility drop of the other two. A fair ranking algorithm should make sure the harm or the utility drop is fairly shared among different groups. 

\begin{table}[t]
\caption{Summary of utility efficiency analysis.}
\label{tab:theory_pareto}
\scalebox{1.0}{
\begin{tabular}{l|lll}
\begin{tabular}[c]{@{}l@{}}The policy \\ optimizing \end{tabular}& \begin{tabular}[c]{@{}l@{}}Ranking policy \\Pareto  efficient\\ for the
users?\end{tabular} & \begin{tabular}[c]{@{}l@{}}Items ranked \\by expected\\relevance within\\ each item group?\end{tabular} & \begin{tabular}[c]{@{}l@{}}Each ranking\\ Pareto efficient\\ for the intents?\end{tabular} \\
\hline
Utility           & \checkmark & \checkmark & \checkmark \\
User fairness & \checkmark & $\times$ & \checkmark \\ 
Item fairness & $\times$ & \checkmark & $\times$ \\
Diversity  & $\times$ & $\times$ & \checkmark 
\end{tabular}}
\end{table}

\section{\algnameNS: Optimizing Rankings for Fairness and Diversity}
\label{sec:ranking_framework} 
Driven by the theoretical analysis from the previous section, we now develop the first ranking algorithm --- called \algname for Two-Sided Fairness and Diversity --- that explicitly accounts for user fairness, item fairness, and diversity requirements. The algorithm proceeds in two steps. In the first step, it optimally satisfies user fairness and item fairness simultaneously through convex optimization. In the second step, the algorithm maximizes diversity subject to the fairness constraints from the first step. 

\subsection{Step 1: Convex Optimization for Item and User Fairness}

In the first step, we optimize the marginal rank probability matrix representation $\mathbf{\Sigma}$ of the ranking policy to satisfy both user and item fairness. As already shown in Section~\ref{sec:rankingsetup}, both user fairness and item fairness can be expressed in terms of $\mathbf{\Sigma}$, which reduces the optimization problem from the exponential space of rankings to the polynomial space of marginal rank probability matrices. This leads to the following convex optimization problem
\begin{equation}
\begin{split}
\mbox{argmax}_{\mathbf{\Sigma}}\ & UF(\mathbf{\Sigma}|q)\\
\mbox{s. t. } &\mathbf{1}^{\top}\mathbf{\Sigma} = \mathbf{1}^{\top},\mathbf{\Sigma}\mathbf{1} = \mathbf{1}, \forall i,j\ 0\leq \mathbf{\Sigma}_{i,j}\leq 1  \\
& \mbox{$\mathbf{\Sigma}$ satisfies item-fairness constraints}
\end{split}
\end{equation}
where $UF$ is the user-fairness objective expressed in terms of the marginal rank probability matrix $\mathbf{\Sigma}$ and $\mathbf{1}$ is the vector of $1$s. We enforce item-fairness in the constraints of the optimization problem, together with the linear constraints that ensure the marginal rank probability matrix $\Sigma$ is doubly stochastic. As long as the user-fairness objective is concave and the item-fairness constraints are linear in $\mathbf{\Sigma}$ (e.g. the disparate treatment constraints in~\eqref{eq:item_fairness_naive}), the problem can be solved 
efficiently and globally optimally with convex optimization algorithms~\cite{boyd2004convex}. 

\subsection{Step 2: Sampling Diverse Rankings}
\label{sec:rank_step_2}
Since we can not directly sample rankings from $\Sigma$, we still need to compute a ranking policy $\pi$ that has $\Sigma$ as its matrix of marginal rank probabilities, and thus the desired user and item fairness. For each matrix $\Sigma$, there are typically many different policies that produce these marginal rank probabilities. Among those policies, we aim to choose the one that provides maximum diversity. This can be formulated as the following optimization problem.
\begin{equation}
\begin{split}
\mbox{argmax}_{\pi}\ & D(\pi|q)\\
\mbox{s. t. } &\mE_{\sigma\sim\pi(\cdot|q)}\left[\mathbbm{1}_{\{\sigma(d_m)=n\}}\right]=\mathbf{\Sigma}_{m,n}\quad\forall m,n
\end{split}
\end{equation}
The constraints in this optimization problem correspond to a Birkhoff-von Neumann (BvN) decomposition~\cite{birkhoff1940lattice}, for which efficient algorithms exist. We present a novel variant of Birkhoff's algorithm~\cite{birkhoff1940lattice} to optimize diversity as illustrated in Algorithm~\ref{alg:BvN}. For each round of Birkhoff's algorithm, we find a permutation (ranking) $\sigma$ that can be sampled from the marginal rank probability matrix $\Sigma$. This corresponds to finding a perfect matching $\sigma$ of the bipartite graph generated from $\Sigma$. Then we add this $\sigma$ to the ranking policy $\pi$ with selection probability to be the smallest entry in the permutation $\sigma$. Then we subtract this selection probability from $\Sigma$ for all the entries in the permutation. The algorithm is proved to be correct and we can always find a perfect matching from the bipartite graph generated from $\Sigma$ in each round~\cite{birkhoff1940lattice}. What is more, the resulting policy consists of no more than $(n-1)^2 + 1$ permutations~\cite{johnson1960algorithm} where $n$ is dimension of $\Sigma$. 

\begin{algorithm}[]
\SetAlgoLined
\SetKwInOut{Input}{input}\SetKwInOut{Output}{output}
 \Input{
 A ranking problem $RP$, 
 A diversity function $g$, \newline
 A marginal rank probability matrix $\mathbf{\Sigma}$
 }
 \Output{
A ranking policy $\pi$
}
{\bf initialization:} {$\forall \sigma$ $\pi(\sigma|q)=0$\newline}
\While{$\mathbf{\Sigma}!=\mathbf{0}$}{
Construct a bipartite graph $G$ with items and positions as vertices and with non-zero elements of $\Sigma$ as edges. \newline
{
$\sigma$ = Local-Search-Match(RP, g, G)
}
$\pi(\sigma|q) = \min_{m}\Sigma_{m,\sigma(d_m)}$\newline
\For {each item $d_m$}{
$\Sigma_{m,\sigma(d_m)} = \Sigma_{m,\sigma(d_m)} - \pi(\sigma|q)$}
}
\Return{$\pi$}
\caption{Greedy Algorithm for BvN Decomposition to Optimize Diversity}
\label{alg:BvN}
\end{algorithm}

\begin{algorithm}[t]
\SetAlgoLined
\SetKwInOut{Input}{input}\SetKwInOut{Output}{output}
 \Input{
 A ranking problem $RP$, 
 a diversity function $g$\newline
 A bipartite graph $G$
 }
 \Output{
A ranking $\sigma$
}
$\sigma^*$ = find a perfect matching of $G$ that maximizes utility. \newline
$Improved = True$\newline
\While{$Improved$}{
$Improved = False$\newline
\For{$d_m$, $d_n$ in $\mathbbm{D}^q\times\mathbbm{D}^q$}{
\If{$(d_n$, $\sigma^*(d_m))$ and $(d_m$, $\sigma^*(d_n))$ $\in$ $G$}{
Construct $\sigma'$ by switching $d_m$ and $d_n$ in $\sigma^*$\newline
\If{$D(\sigma '|q)>D(\sigma^*|q)$}{
$\sigma^*$ = $\sigma'$ and 
$Improved = True$}
}
}
}
\Return{$\sigma^*$}
\caption{Local-Search-Match(RP, g, G)}

\label{alg:local_search_match}
\end{algorithm}

With the additional goal in the objective of constructing a policy that maximizes diversity, we choose the permutation matrices in each step greedily to maximize diversity as detailed in Algorithm~\ref{alg:local_search_match}. Since finding the permutation with the largest diversity that is satisfiable in $\Sigma$ is NP-hard, we start with the permutation that maximizes utility among the ones that satisfy the conditions of the BvN decomposition. This can be solved by polynomial-time minimum-cost perfect matching algorithms~\cite{kleinberg2006algorithm}. 
We then adopt a local search strategy that switches two items if the switch increases diversity. We also tried more expensive search strategies that exhaustively search up to position $3$ and found the difference to be small. 
We present the details of the other search strategy in the appendix.

Note that maximizing diversity reduces the utility variance to the users across the rankings drawn from $\pi$. This can be seen as a form of envy reduction~\cite{brams1996fair,steihaus1948problem,procaccia2013cake}, where envy measures the individual reduction in utility that a particular user experiences by not drawing the user's optimal-utility ranking from $\pi$. To show this, we derive an upper bound $D_{UB}^{\Sigma^{\pi,q}}$ of the diversity as a function of $\Sigma^{\pi,q}$
\begin{equation}
\begin{split}
D(\pi|q)& = \mE_{\sigma\sim\pi(\cdot|q),i\sim\mathcal{I}_{\mathcal{U}^q}}\left[g(U(\sigma|i,q))\right]\\
&\leq \mE_{i\sim\mathcal{I}_{\mathcal{U}^q}}\left[g\left(\mE_{\sigma\sim\pi(\cdot|q)}[U(\sigma|i,q)]\right)\right]\\
&= \mE_{i\sim\mathcal{I}_{\mathcal{U}^q}}\left[g\left(\left(\mathbf{r}^{i}\right)^{\top}\Sigma^{\pi,q}\mathbf{e}\right)\right]= D_{UB}^{\Sigma^{\pi,q}}, 
\end{split}
\end{equation}
where $\mathbf{r}^{i}$ is the relevance vector to the intent $i$ with $\mathbf{r}^{i}_{m} = r(d_m, i)$. The equality holds when, for each user with a particular intent, the utility for that intent is the same across the rankings sampled from the ranking policy --- which means that there is no envy of a user that receives a particular ranking to the other rankings that could have been sampled from the ranking policy.

Note that the upper bound is determined by $\Sigma^{\pi,q}$, which is optimized in the first step. The second step maximizes diversity to match this upper bound, which can be interpreted as reducing the envy of the users. This also illustrates a value judgment in the design of \algnameNS, where we optimize user and item fairness as the primary criteria, and diversity as a secondary one. This is also reflected in \autoref{tab:theory_zero_util}, where \algname is shown to guarantee non-zero utility to the user and item groups, but not necessarily to cover the maximum amount of intent. 

\section{Empirical Evaluation}
\label{sec:exp}
In addition to the theoretical characterizations, we now evaluate empirically in how far different ranking algorithms affect utility, user fairness, item fairness and diversity on a movie recommendation dataset. 

\subsection{Dataset}

We constructed the first benchmark dataset that provides intent, user group, and item group annotations. We collected $100$ movies from different genres \{ Romance (20), Comedy (25), Action (25), Thriller (15), Sci-Fi (15) \} that are lead by actors of different races \{black-lead (20), white-lead (80)\}. We treat the genres as the intent set and the leading-actor races as the item group set. The relevance of a movie to a genre is the average user rating on IMDB\footnote{https://www.imdb.com/} if the movie belongs to that genre and 0 otherwise. 
To fully leverage the range of the ratings, we subtract the minimum rating $6$ in the dataset from all the ratings to obtain the relevance. For the users, we regard male and female as two user groups and set the user proportion $\rho_{male}\in [0,1]$ as 0.6 by default and $\rho_{female} = 1-\rho_{male}$. To enable varying the intent similarity between the two user groups, we arbitrarily construct two dissimilar intent distributions $\mathcal{I}_{1} = [0.5,0.5,0, 0, 0]$ and $\mathcal{I}_2 = [0, 0, 0.5, 0.25, 0.25]$ over the five genres. We use an intent similarity factor $s\in [0,1]$ (0.5 by default) to control the intent similarity between the two user groups $\mathcal{I}_{male} = (1-0.5s)\mathcal{I}_{1} + 0.5s\mathcal{I}_{2}$ and $\mathcal{I}_{female} = (1-0.5s)\mathcal{I}_{2} + 0.5s\mathcal{I}_{1}$. 

\subsection{Experiment Setup} 
All results are averaged over 5,000 samples (50,000 samples for the results in \autoref{tab:exp_main}), where each sample consists of 15 randomly selected movies to be ranked. To make sure the inputs to the merit and diversity functions are within their domains for all the algorithms while there is a clear trade-off between the policies, we set the user fairness function as $f(\cdot) = log (\cdot - 0.6)$ and the diversity function as $g(\cdot) =log(\cdot + 0.0001)$. To control the exposure steepness, we set the exposure function as $e(\cdot) = (\frac{1}{\cdot})^\eta$ where $\eta$ controls the exposure steepness and we set $\eta = 1$ by default. For all the experiments, we use the default parameters introduced in this section unless explicitly stated otherwise.  

We compare \algname with 4 other policies that maximize utility, item fairness, user fairness, and diversity respectively. For \algname and the policies that maximize item fairness and user fairness, we first satisfy the fairness goals through convex optimization\footnote{We use MOSEK (https://www.mosek.com/) to solve the convex optimization problem. }, and then we optimize the diversity by running the greedy BvN decomposition algorithm. We use the greedy submodular optimization approximation algorithm with two matroid constraints to maximize diversity~\cite{krause2014submodular}. The algorithm is detailed in the appendix. To avoid cases where the item-fairness constraints can not be satisfied, we optimize the one-sided disparate treatment constraints proposed in ~\cite{Singh/Joachims/19a} in the experiments: $\frac{E(\pi|DG_2)}{M(DG_2,q)}\leq \frac{E(\pi|DG_1)}{M(DG_1,q)}$ with $M(DG_1, q) \leq M(DG_2,q)$. 

For clarity of presentation, we bring the user fairness, the diversity and the upper bound on the diversity calculated from the marginal rank probability matrix (diversity UB) on the same scale as the overall utility by applying the inverse of user fairness function and diversity function to each of them to get $f^{-1}(UF(\pi|q))$, $g^{-1}(D(\pi|q))$, and $g^{-1}(D^{\Sigma^{\pi,q}}_{UB})$, where $f^{-1}(\cdot) = e^{\cdot} + 0.6$ and $g^{-1}(\cdot) = e^{\cdot} - 0.0001$. Item unfairness is the amount of violation of the one-sided disparate treatment constraints $\text{max}(0, \frac{E(\pi|DG_2)}{M(DG_2,q)} - \frac{E(\pi|DG_1)}{M(DG_1,q)} )$ with $M(DG_1, q) \leq M(DG_2,q)$. 

\begin{figure*}[ht]
  \centering
    \includegraphics[width=0.99\textwidth]{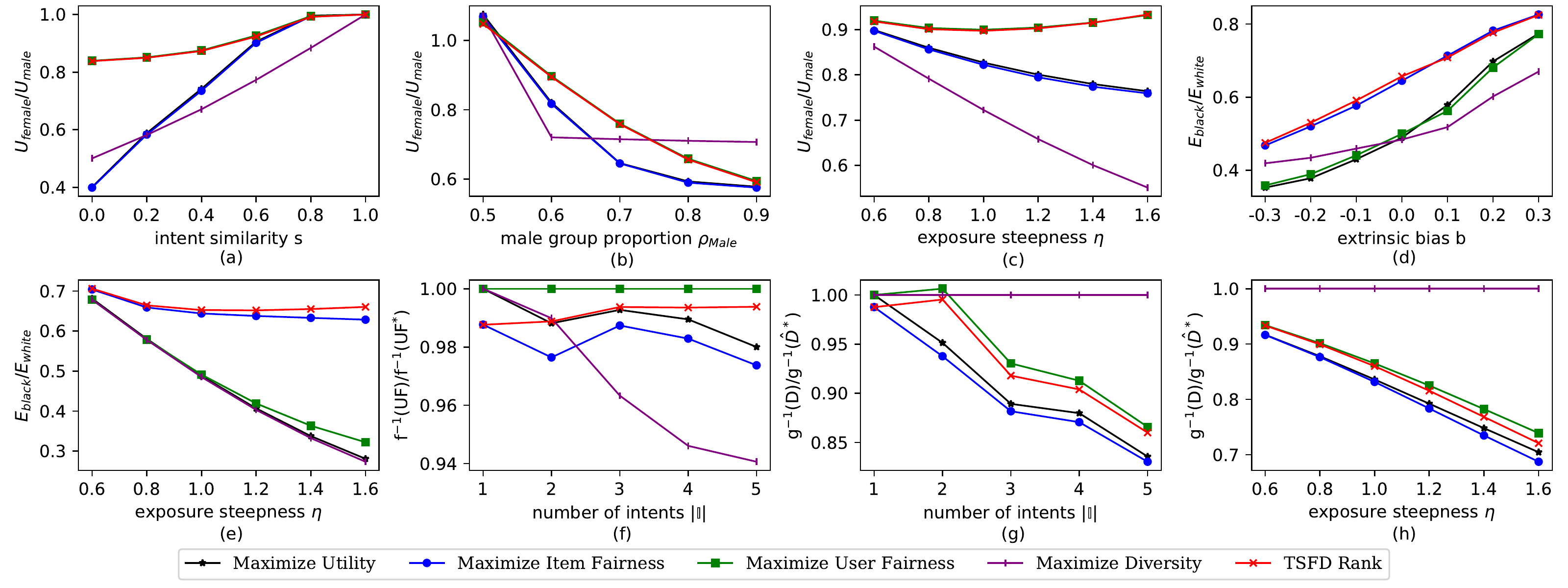}
  \caption{The effects of user intent similarity, user group proportion, extrinsic bias to an item group, exposure steepness and number of intents on the item groups, user groups and intents. 
  }
 \label{fig:exp_comprehensive}
\end{figure*}
\subsection{Empirical Results}
\textbf{How do different methods trade-off between user fairness, item fairness, diversity and utility?}
We show the empirical results with the default setup in \autoref{tab:exp_main}. As expected, the ranking algorithms that consider only one of the measures excel at that measure but achieve sub-optimal performance on the other ones. In contrast, the proposed \algname explicitly controls all desiderata by sacrificing some utility to achieve perfect item fairness, second-best user fairness and third-best diversity (very close to the second-best). The diversity upper bound provides a skyline of how much diversity \algname can possibly achieve. So the small difference between the diversity achieved by the policy that maximizes item fairness, the policy that maximizes user fairness, the policy produced by \algname and their respective diversity upper bound shows that the greedy BvN decomposition algorithm achieves diversity very close to the upper bound.

\begin{table}[t]
\caption[Caption for LOF]{performance of different ranking algorithms\footnotemark}
\label{tab:exp_main}
\scalebox{0.95}{
\begin{tabular}{l|lllll}
\begin{tabular}[c]{@{}l@{}}The policy\\ optimizing\end{tabular}&Utility  & \begin{tabular}[c]{@{}l@{}}Item\\ unfairness\end{tabular}     & \begin{tabular}[c]{@{}l@{}}User\\ fairness\end{tabular}   & Diversity & \begin{tabular}[c]{@{}l@{}}Diversity\\ UB\end{tabular}\\
        \hline
Utility & \textbf{1.518} & 0.186 & 1.447 & 1.016 & 1.016 \\
Item fairness & 1.509 & \textbf{0.000} & 1.437 & 1.010 & 1.013  \\
User fairness & 1.498 & 0.193 & \textbf{1.476} & 1.052 & 1.062  \\
Diversity & 1.428 & 0.185 & 1.390 & \textbf{1.214} & 1.214 \\
\algnameNS & 1.489 & \textbf{0.000} & 1.466 & 1.045 & 1.055 
\end{tabular}
}
\end{table}
\footnotetext{The standard error of each value presented in the table is smaller than 0.001.}

\textbf{How do user intent similarity, user group proportion, and exposure steepness affect user fairness?}
\autoref{fig:exp_comprehensive} (a) (b) (c) show the effect of the three factors on the utility ratio between female and male user groups $U_{female}/U_{male} =U(\pi|female,q)/U(\pi|male,q)$, which measures the utility difference between the two user groups. 
For the policy that maximizes user fairness, the minority (female) group gets a smaller ratio of utility as the intent similarity decreases. The ratio also decreases as the male group proportion increases, and it stays flat with varying exposure steepness. This is expected since the user fairness objective gives larger weight for the majority group but is oblivious to exposure steepness. The proposed \algname achieves almost the same ratio as the policy maximizing only user fairness, which shows its effectiveness on fairly distributing the utility drop due to the other desiderata between the two user groups. The policies that maximize item fairness or overall utility amplify the utility drop of the minority (female) user group more than \algnameNS. The policies that maximize diversity sometimes amplify the utility drop while sometimes over-correcting it.   

\textbf{How do extrinsic bias and exposure steepness affect item fairness?}
Biased relevance estimates, which might come from biased data, can contribute to unfair exposure allocation to the items~\cite{Singh/Joachims/18a}. To simulate the bias, for each black-lead movie $d$ of genre $i$, we set the biased relevance as $r_b(d, i) = (1+b)r(d,i) $ where $b$ is the bias level to the black-lead movies. The results with varying biases are shown in \autoref{fig:exp_comprehensive} (d). The policy maximizing item fairness ensures roughly a linear change in exposure ratio as the bias increases, which is expected since the exposure ratio is a linear function of the average relevance of black-lead movies, which in turn is a linear function of the bias level $b$. The proposed \algname achieves similar exposure ratio as the policy maximizing item fairness, while all the other methods lead to undesirable over-amplifications of the bias towards the less represented black-lead movies. 
\autoref{fig:exp_comprehensive} (e) shows that when the exposure steepness increases, both \algname and the policy maximizing item fairness manage to control the winner-takes-all dynamics while all the other methods fail to ensure a more equitable amount of exposure to the less represented black-lead movies. 

\textbf{How do the number of intents and exposure steepness affect diversity?}
The diversity ratio $g^{-1}(D)/g^{-1}(\hat{D}^{*}) = \frac{g^{-1}(D(\pi|q))}{g^{-1}(\hat{D}^*)}$ and the user fairness ratio $f^{-1}(UF)/f^{-1}(UF^{*}) = \frac{f^{-1}(UF(\pi|q))}{f^{-1}(UF^*)}$ measure how far a policy deviates from the policies that optimize each desideratum where $UF^*$ is the user fairness achieved by the policy maximizing user fairness and $\hat{D}^*$ is the diversity achieved by optimizing diversity by the greedy submodular approximation algorithm. \autoref{fig:exp_comprehensive} (f) shows that as the number of intents gets larger, maximizing diversity gets further away from maximizing user fairness. 
\autoref{fig:exp_comprehensive} (g) and (h) show that as the number of intents gets larger and as the exposure distribution gets steeper, the policies that satisfy other desiderata deviate further from the policy maximizing diversity. Combined with the other empirical findings, these results show that maximizing diversity fails to achieve user or item fairness and vice versa. That \algname achieves the third-best diversity is expected, since it prioritizes fairness over diversity and only considers diversity in the second step when the marginal rank probability matrix representation $\Sigma^{\pi,q}$ of the ranking policy with a sub-optimal diversity upper bound is already determined. 

\section{Conclusion}
We analyzed the interplay between user fairness, item fairness and diversity for rankings in two-sided markets and found that they are three independent and conflicting goals. Driven by the analysis, we proposed \algnameNS, the first ranking algorithm that explicitly enforces user fairness, item fairness and diversity. \algname can optimally satisfy user fairness and item fairness through convex optimization and then optimize diversity subject to the fairness constraints via a novel BvN decomposition algorithm. Empirical results on a movie recommendation dataset confirm that \algname can effectively and robustly control the three desiderata. 

\begin{acks}
This research was supported in part by NSF Awards IIS-1901168 and IIS-2008139. All content represents the opinion of the authors, which is not necessarily shared or endorsed by their respective employers and/or sponsors.
\end{acks}

\bibliographystyle{ACM-Reference-Format}
\balance 
\bibliography{paper}


\appendix
\onecolumn
\section{Greedy Birkhoff-von Neumann Decomposition Algorithm to Sample Diverse Rankings}
\label{sec:TSDF_second_step}

In addition to the local search strategy in Algorithm~\ref{alg:local_search_match} for finding a permutation from the marginal rank probability matrix that maximizes diversity, we also design an exhaustive search strategy illustrated in Algorithm~\ref{alg:exhaustive_search_match}. The algorithm exhaustively search all the possible allocations of the top $l$ positions and allocate the other positions arbitrarily. The algorithm selects the one perfect matching with the largest diversity. 

We show the comparison of different strategies under the default setup described in Section~\ref{sec:exp} on the movie dataset in \autoref{tab:exp_bvn}.
ES denotes exhaustive search, LSI denotes local search with utility maximizing ranking initialization and LSNI denotes local search without utility maximizing ranking initialization. The diversity upper bound might not be achieved by any ranking policy since we are constrained to sampling rankings from the marginal rank probability matrix. Nevertheless, all algorithms still achieve diversity very close to the diversity upper bound. There is little difference between different strategies. 
% We adopted local search with initialization in all the other experiments. 

\begin{table}[h]
\caption{Comparison of different BvN decomposition algorithms}
\label{tab:exp_bvn}
\scalebox{1.0}{
\begin{tabular}{l|llllll|l}
 &ES level 0 (Random)  & ES level 1 & ES level 2  & ES level 3 & LSI & LSNI & Upper Bound \\
 \hline
Diversity & 1.04421 & 1.04430 & 1.04430 & \textbf{1.04431} & 1.04430 & 1.04415 & 1.05531 \\
\end{tabular}
}
\end{table}

\begin{algorithm}[]
\SetAlgoLined
\SetKwInOut{Input}{input}\SetKwInOut{Output}{output}
 \Input{
 A ranking problem $RP$, 
 A diversity function $g$, 
 A bipartite graph $G$, 
 The exhaustive search level $l$. 
 }
 \Output{
A ranking $\sigma$
}
{\bf initialization:} {$\sigma^*=None$, $D^*=-\infty$, $N = |\mathbbm{D}^q|$\newline}
\eIf{$l = 0$}{
\Return {an arbitrary perfect matching $\sigma$ of $G$.}}
{
Construct a sub-graph $SG$ of $G$ with all the items and the positions in $\{1,2,...,l\}$ as vertices and the edges between them.\newline
Find all maximum matchings $M_{SG}$ of $SG$.\newline
\For{each matching $M$ in $M_{SG}$ }{
Construct a sub-graph $SG'$ of $G$ by deleting vertices in $M$ from $G$.\newline
Find a maximum matching $M'$ of $SG'$.\newline
\If{the number of edges $|M'|$ in $M$ is $N-l$ }{
Construct a ranking $\sigma$ by combining $M$ and $M'$. \newline
\If{$D(\sigma|q)$ > $D^*$}{
$D^*,\sigma^* =D(\sigma|q),\sigma$}
}
}
\Return {$\sigma^*$}
}
\caption{Exhaustive-Search-Match(RP,g, G,l)}
\label{alg:exhaustive_search_match}
\end{algorithm}

\section{Greedy Approximation Algorithm for Submodular Optimization of Diversity}
\label{sec:submodular_diversity}
A ranking that maximizes diversity can constitute a deterministic ranking policy that maximizes the expected diversity. The problem of finding a ranking that maximizes diversity can be formulated as follows.  Denote the set of all positions as $\mathbbm{L} = \{1, 2, 3,..., N\}$ and all assignments of one item to a position as $\mathbbm{A} = \mathbbm{D}\times\mathbbm{L}$. The diversity optimization problem is
\begin{equation}
\begin{split}
    \mbox{argmax}_{\sigma\subseteq\mathbbm{A}}\ & D(\sigma|q)\\
    \mbox{s. t. } & \mbox{Each product shows up at most once in $\sigma$.}\\
    & \mbox{Each position shows up at most once in $\sigma$.}\\
\end{split}
\end{equation}

$D(\cdot|q):2^{\mathbbm{A}}\rightarrow \mathbbm{R}$ is monotone submodular. The two constraints are two partition matroid constraints. The problem is generally NP-hard but a simple greedy approximation algorithm enjoys $\frac{1}{3}D^{*}$ performance guarantee and works much better than that in practice~\cite{krause2014submodular}. The gredy algorithm executes as follows: starting with the empty set of $\sigma_0=\{\}$, the greedy algorithm greedily selects an item-position pair in $\mathbbm{A}$ that does not violate the two constraints but maximizes diversity 
\begin{equation}
    \sigma_m:= \sigma_{m-1}\cup\{\argmax_{a\in\mathbbm{A},\sigma_{m-1}\cup\{a\} \mbox{satisfies the two constraints}}[D(\sigma_{m-1}\cup\{a\}|q)]\}.
\end{equation}

The resulting assignment set $\sigma_{|\mathbbm{D}^q|}$ constitutes a ranking that has diversity at least $\frac{1}{3}D^{*}$.

\section{Theorems}
\label{sec:theorems}

We provide the proofs of all the theorems in the main paper in this section. For clarity, we present some related definitions here. 

\begin{definition}{(Zero/non-zero utility under a ranking)}
Denote $G$ as an item group, a user group or an intent. For a non-degenerate ranking problem, $G$ has non-zero utility under ranking $\sigma$ if $U(\sigma|G,q)>0$ and has zero utility under ranking policy $\sigma$ if $U(\sigma|G,q)=0$.
\end{definition}

\begin{definition}{(Zero/non-zero utility under a ranking policy)}
Denote $G$ as an item group, a user group or an intent. For a non-degenerate ranking problem, $G$ has non-zero utility under ranking policy $\pi$ if $U(\pi|G,q)>0$ and has zero utility under ranking policy $\pi$ if $U(\pi|G,q)=0$.
\end{definition}

\begin{definition}{(Non-zero optimal utility under a ranking policy)}
Denote $G$ as an item group, a user group or an intent. For a non-degenerate ranking problem, $G$ has non-zero optimal utility if there exists a ranking policy $\pi$ such that $G$ has non-zero utility under $\pi$. 
\end{definition}

\begin{definition}{(Satisifiable Disparate treatment constraints)}
For a non-degenerate ranking problem and a merit function $M$, we say the disparate treatment constraints with merit function $M$ is satisfiable for the ranking problem if there exists a ranking policy $\pi$ that satisfies the disparate treatment constraints for the ranking problem. 
\end{definition}

As~\cite{Singh/Joachims/19a} pointed out, there exists degenerate cases where the disparate treatment constraints are not satisfiable for a ranking problem.

\subsection{Maximizing Utility}
\subsubsection{Zero-utility Analysis of maximizing utility}

\begin{proof}[Proof of Theorem~\ref{theo:util_zero_user}]
\label{proof:util_zero_user}
Proof by construction. 
We first construct a non-degenerate ranking problem as illustrated in \autoref{fig:example}. $\mathbbm{I} = \{i_1, i_2, i_3\}$, $\mathbbm{UG} = \{UG_1, UG_2\}$, $\rho_{UG_1} = 0.5$ and $\rho_{UG_2} = 0.5$, $\mathcal{I}_{\mathcal{UG}_1}(i_1)= 0.6$, $\mathcal{I}_{\mathcal{UG}_1}(i_2)= 0.4$, $\mathcal{I}_{\mathcal{UG}_2}(i_3)= 1$,  $\mathcal{I}_{\mathcal{UG}_1}(i_3) = \mathcal{I}_{\mathcal{UG}_2}(i_1) = \mathcal{I}_{\mathcal{UG}_2}(i_2) = 0$,  $\mathbbm{UG} = \{DG_1, DG_2\}$ $d_{m*} = \{d_{m1}, d_{m,2}, d_{m,3}\}$ for $m\in\{1,2,3,4,5,6\}$.  $\mathbbm{D}^q= \cup_{m} d_{m*}$,  $DG_1 = d_{1*}\cup d_{2*}\cup d_{3*}$, $DG_2 = d_{4*}\cup d_{5*}\cup d_{6*}$.  $r(d_{1*}, i_{1}) = r(d_{2*}, i_{2}) = 0$. We denote with $r(d_{m*},i) = \cdot$ as all the items in $d_{m*}$ has relevance $\cdot$ to intent $i$. $r(d_{1*}, i_{1}) = r(d_{2*},i_{2}) = r(d_{3*}, i_{3}) = 1$, $r(d_{4*}, i_{1}) = r(d_{5*},i_{2}) = r(d_{6*}, i_{3}) = 0.9$. All the other item-intent pairs have zero relevance. The exposure of the first three positions have exposure $1$, i.e. $e(1) = e(2) = e(3) = 1$. All the other positions have exposure $0$. We will show that $UG_1$ has zero utility in any utility maximizing policy $\pi$ while it has non-zero optimal utility.

First, $UG_1$ has non-zero optimal utility since any ranking that allocates any one item in $d_{1*}$, $d_{2*}$, $d_{4*}$ or $d_{5*}$ in a non-zero exposure position will have non-zero utility for $UG_1$. 

Second, in any policy $\pi$ that maximizes utility, items in $d_{3*}$ take all the exposure and all the other items have zero exposure. Since items in $d_{3*}$ have the largest expected relevance $\mE_{i\sim\mathcal{I}_{\mathcal{U}}[r(d,i)]}$ than the other items, they occupy the $3$ non-zero exposure positions in any ranking sampled from $\pi$. If not, we can switch the items to get a ranking with larger utility which contradicts the assumption that $\pi$ maximizes utility. 

Since only items in $d_{1*}$, $d_{2*}$, $d_{4*}$, $d_{5*}$  have non-zero relevance to $UG_1$ but they all get $0$ exposure, the utility for $UG_1$ is zero. 
\end{proof}

\begin{theorem}{(Maximizing overall utility can lead to $0$ utility for an item group.)}
\label{theo:util_zero_item}
There exists a non-degenerate ranking problem such that the policy $\pi$ maximizing the overall utility has utility $U(\pi|DG,q) = 0$ for an item group $DG$. 
\end{theorem}

\begin{proof}
\label{proof:util_zero_item}
Proof by construction. We still use the example in \autoref{fig:example} that was detailed in the Proof of Theorem~\ref{theo:util_zero_user}. We will show that $DG_2$ has zero utility in any utility maximizing policy $\pi$ while it has non-zero optimal utility.

First, $DG_2$ has non-zero optimal utility since any ranking that allocates an item in $DG_2$ in a non-zero exposure position will have non-zero utility for $DG_2$.  
Second, as shown in the Proof of Theorem~\ref{theo:util_zero_user}, in any policy $\pi$ that maximizes utility, items in $d_{3*}$ take all the exposure and all the other items have zero exposure. Since items in $DG_2$ have zero exposure, the utility for $DG_2$ is zero. 
\end{proof}

\begin{theorem}{(Maximizing overall utility can not ensure each ranking sampled from the policy covers the maximum amount of intent covered by any ranking.)}
\label{theo:util_zero_max_intent}
There exists a non-degenerate ranking problem such that the amount of intent covered by any ranking sampled from the policy maximizing the overall utility is less than the maximum amount of intent covered by any ranking. 
\end{theorem}

\begin{proof}
\label{proof:util_zero_max_intent}
Proof by construction. We still use the example in \autoref{fig:example} that was detailed in the Proof of Theorem~\ref{theo:util_zero_user}. The maximum amount of intent covered by any ranking is $1$ since we can select $3$ items relevant to different intents and put them in the $3$ non-zero exposure positions. We will show that the policy maximizing utility will only cover $i_3$. As shown in the Proof of Theorem~\ref{theo:util_zero_user}, in any policy $\pi$ that maximizes utility, items in $d_{3*}$ take all the exposure and all the other items have zero exposure. Since items relevant to $i_1$ and $i_2$ have zero exposure, $i_1$ and $i_2$ have zero utility. The amount of intent covered by the policy maximizing utility is only $0.5$ which is less than $1$. 
\end{proof}

\subsubsection{Efficiency Analysis of maximizing utility}
\begin{theorem}{(The policy maximizing the overall utility is Pareto efficient for the user groups.)}
\label{theo:util_efficiency_user}
For any non-degenerate ranking problem, if a ranking policy $\pi$ maximizes the overall utility, then $\pi$ is Pareto efficient for the user groups.  
\end{theorem}
\begin{proof}
Proof by contradiction. Assume there exists a ranking policy $\pi$ such that $\pi$ maximizes the utility i.e. $\forall \pi'$ $U(\pi'|q)\leq U(\pi|q)$ and $\pi$ is not Pareto efficient for the user groups. By the definition of Pareto efficiency, there exists a ranking policy $\pi''$ such that $\pi''$ dominates $\pi$ for the user groups. By the definition of dominance,
\begin{displaymath}
    U(\pi''|q) - U(\pi|q) =\sum_{UG\in\mathbbm{UG}}\rho_{UG}^q(U(\pi''|UG,q)-U(\pi|UG,q))>0
\end{displaymath}
which contradicts the assumption that $\pi$ maximizes the utility. 
\end{proof}

\begin{theorem}{(Items are ranked by their expected relevance within each item group in the policy maximizing the overall utility.)}
\label{theo:util_efficiency_item}
For any non-degenerate ranking problem, if a ranking policy $\pi$ maximizes the overall utility, then items are ranked by their expected relevance to the whole user population within each item group under $\pi$.
\end{theorem}
\begin{proof}
\label{proof:util_efficiency_item}
If there is a ranking $\sigma$ with $\pi(\sigma|q)>0$ and items are not ranked by  their expected relevance within each item group, we can construct a new ranking $\sigma'$ by switching one pair of items that are not ranked by their expected relevance within each item group so that $U(\sigma'|q)>U(\sigma|q)$. We can replace $\sigma$ with $\sigma'$ in the ranking policy $\pi$ to construct a new ranking policy $\pi'$ with larger overall utility. This contradicts that $\pi$ maximizes utility. 
\end{proof}

\begin{theorem}{(Maximizing utility ensures each ranking sampled from the policy is Pareto efficient for the intents.)}
\label{theo:util_efficiency_ranking_intent}
For any non-degenerate ranking problem, if a ranking policy $\pi$ maximizes the overall utility, then for any ranking $\sigma$ with $\pi(\sigma|q)>0$, $\sigma$ is Pareto efficient for the intents.  
\end{theorem}
\begin{proof}
\label{proof:util_efficiency_ranking_intent}

First, each ranking sampled from the policy $\pi$ maximizing utility is a ranking maximizing utility. Otherwise, we can replace it with a ranking maximizing utility in the ranking policy to get a policy with larger utility. We only need to prove that any ranking that maximizes utility is Pareto efficient for the intents among the rankings. 

Proof by contradiction. Assume there exists a ranking $\sigma$ such that $\sigma$ maximizes the utility i.e. $\forall \sigma'$ $U(\sigma'|q)\leq U(\sigma|q)$ and $\sigma$ is not Pareto efficient for the intents. By the definition of Pareto efficiency, there exists a ranking $\sigma''$ such that $\sigma''$ dominates $\sigma$  for the intents. By the definition of dominance,
\begin{displaymath}
    U(\sigma''|q) - U(\sigma|q) =\mE_{i\sim\mathcal{I}_{\mathcal{U}^q}}\bigg[U(\sigma''|i,q) - U(\sigma|i,q)\bigg]>0
\end{displaymath}
which contradicts the assumption that $\sigma$ maximizes the utility. 
\end{proof}

\subsection{Maximizing user fairness}

\subsubsection{Zero-utility analysis of maximizing user fairness.}

\begin{proof}[Proof of Theorem~\ref{theo:uf_zero_user}]
\label{proof:uf_zero_user}
First, there always exists a ranking policy that has non-zero utility for each user group (e.g. the mixture of the ranking policies that maximize utility for each user group). 

Second, for any ranking policy $\pi$ that has non-zero utility for every user group, we can construct a user fairness function $f$ such that any ranking policy $\pi'$ where at least one user group has zero utility has smaller user fairness than $\pi$. Without loss of generality, we denote $UG_1$ one of the user groups that has zero utility under $\pi'$. Denote the maximum utility of a user group over all user groups and all ranking policies as $U_{max}^{\mathbbm{UG}}=argmax_{UG\in\mathbbm{UG},\pi''}U(\pi''|UG,q)$ and the minimum utility of a user group across all user groups under policy $\pi$ as $U_{min,\pi}^{\mathbbm{UG}}=argmin_{UG\in\mathbbm{UG}}U(\pi|UG,q)$. The increasing concave function $f$ is constructed as a piece-wise linear function

\[ f(\cdot)=\begin{cases} 
      k_1(\cdot - t_1)  \quad\mbox{if}\quad\cdot \leq t_1  \\
      k_2(\cdot - t_1)  \quad\mbox{if}\quad\cdot > t_1
  \end{cases}
\]
with $k_1$ $k_2$ $t_1$ to be set later. We want to show that $UF(\pi|q)>UF(\pi'|q)$

\begin{equation}
\begin{split}
    &UF(\pi|q) - UF(\pi'|q)\\
=&\sum_{UG\in\mathbbm{UG}}\rho_{UG}f(U(\pi|UG,q))-\sum_{UG\in\mathbbm{UG}}\rho_{UG}f(U(\pi'|UG,q))\\
=&\rho_{UG_1}[f(U(\pi|UG_1,q))-f(U(\pi'|UG_1,q))] +\sum_{UG\in\mathbbm{UG} \setminus UG_1}\rho_{UG}[f(U(\pi|UG,q))-f(U(\pi'|UG,q))]\\
\geq&\rho_{UG_1}[f( U_{min,\pi}^{\mathbbm{UG}} )-f(0)] +\sum_{UG\in\mathbbm{UG}\setminus UG_1}\rho_{UG}[f( U_{min,\pi}^{\mathbbm{UG}} )-f(U_{max}^{\mathbbm{UG}})]\\
\geq&\rho_{UG_1} U_{min,\pi}^{\mathbbm{UG}}k_1  + (1-\rho_{UG_1}) (U_{min,\pi}^{\mathbbm{UG}} - U_{max}^{\mathbbm{UG}} )k_2\quad (\mbox{assume }  t_1 = U_{min,\pi}^{\mathbbm{UG}}, k_2 > 0, k_1 > 0)
\end{split}
\end{equation}

We set $t_1 =U_{min,\pi}^{\mathbbm{UG}}$ $k_2$ to be any positive real number, $k_1 > max(k_2, \frac{ (1-\rho_{UG_1})( U_{max}^{\mathbbm{UG}}  - U_{min,\pi}^{\mathbbm{UG}}) }{\rho_{UG_1}U_{min,\pi}^{\mathbbm{UG}}}k_2)$ so that $UF(\pi|q) > UF(\pi'|q)$. 

Since if a ranking policy has zero utility for a user group then the ranking policy has smaller user fairness than $\pi$ using user fairness function $f$, every user group has non-zero utility in the policy maximizing user fairness. 
\end{proof}

\begin{theorem}{(Maximizing user fairness can lead to zero utility for an item group.)}
\label{theo:uf_zero_item}
There exists a non-degenerate ranking problem such that for any user fairness function $f$, the policy $\pi$ maximizing user fairness has utility $U(\pi|DG,q)=0$ for an item group $DG$. 
\end{theorem}

\begin{proof}
\label{proof:uf_zero_item}
Proof by construction. We use the example in \autoref{fig:example} that was detailed in the Proof of Theorem~\ref{theo:util_zero_user}. We will show that $DG_2$ has zero utility under the policy maximizing user fairness no matter what increasing concave user fairness function $f$ we choose. 

For any increasing concave user fairness function $f$, the policy maximizing user fairness will always rank items in $d_{1*}$ over items in $d_{2*}$, $d_{4*}$, $d_{5*}$ since items in $d_{1*}$ have strictly larger expected relevance to user group $UG_1$ and all the four sets have zero utility for $UG_2$. If not, we can switch the items to get larger user fairness. Similarly, the policy maximizing user fairness will rank items in $d_{3*}$ over items in $d_{6*}$. Thus, items in $d_{1*}$, $d_{3*}$ will occupy all the non-zero exposure positions. And all the other items get 0 exposure. Since items in $DG_2$ get 0 exposure, the policy has zero utility for $DG_2$. 

\end{proof}

\begin{theorem}{(Maximizing user fairness can not ensure each ranking sampled from the policy maximizing user fairness covers the maximum amount of intent covered by any ranking.)}
\label{theo:uf_zero_max_intent}
There exists a non-degenerate ranking problem such that for any user fairness function $f$, the amount of intent covered by any ranking sampled from the policy maximizing user fairness is less than the maximum amount of intent covered by any ranking. 
\end{theorem}

\begin{proof}
\label{proof:uf_zero_max_intent}
Proof by construction. We use the example in \autoref{fig:example} that was detailed in the Proof of Theorem~\ref{theo:util_zero_user}. The maximum amount of intent covered by any ranking is $1$, since we can select $3$ items relevant to different intents and put them in the $3$ non-zero exposure positions. We will show that the policy maximizing user fairness will only cover $i_1$ and $i_3$. As shown in the Proof of Theorem~\ref{theo:uf_zero_item},  for any increasing concave user fairness function $f$, the policy maximizing user fairness will always rank items in $d_{1*}$ over items in $d_{2*}$, $d_{4*}$, $d_{5*}$. Thus items in $d_{2*}$ and $d_{5*}$ get zero exposure. So intent $i_2$ has zero utility.  As a result, any ranking will cover less intent than $1$, the maximum amount of intent covered by any ranking. 
\end{proof}

\subsubsection{Efficiency analysis of maximizing user fairness}

\begin{proof}[Proof of Theorem~\ref{theo:uf_efficiency_user}]
\label{proof:uf_efficiency_user}
Proof by contradiction. Assume there exists a ranking policy $\pi$ such that $\pi$ maximizes the user fairness i.e. $\forall \pi'$ $UF(\pi'|q)\leq UF(\pi|q)$ and $\pi$ is not Pareto efficient for the user groups. By the definition of Pareto efficiency, there exists a ranking policy $\pi''$ such that $\pi''$ dominates $\pi$ for the user groups. By the definition of dominance and the increasing property of $f$,
\begin{displaymath}
    UF(\pi''|q) - UF(\pi|q) =\sum_{UG\in\mathbbm{UG}}\rho_{UG}^q(f(U(\pi''|UG,q))-f(U(\pi|UG,q)))>0
\end{displaymath}
which contradicts the assumption that $\pi$ maximizes the user fairness. 
\end{proof}

\begin{theorem}{(Maximizing user fairness ensures each ranking sampled from the policy is Pareto efficient for the intents.)}
\label{theo:uf_efficiency_ranking_intent}
For any non-degenerate ranking problem and any user fairness function $f$, if a ranking policy $\pi$ maximizes user fairness, then for any ranking $\sigma$ with $\pi(\sigma|q)>0$, $\sigma$ is Pareto efficient for the intents. 
\end{theorem}
\begin{proof}
\label{proof:uf_efficiency_ranking_intent}
Proof by contradiction. For any policy $\pi$ that maximizes user fairness, assume there exists a ranking $\sigma$ such that $\pi(\sigma|q)>0$ while $\sigma$ is not Pareto efficient for the intents. By the definition of Pareto efficiency, there exists a ranking $\sigma'$ that dominates $\sigma$ for the intents. We construct a new ranking policy $\pi'$ by replacing $\sigma$ with $\sigma'$. By the definition of dominance and the increasing property of $f$,

\begin{displaymath}
\begin{split}
    &UF(\pi'|q) - UF(\pi|q)\\ =&\sum_{UG\in\mathbbm{UG}}\rho_{UG}^q(f(U(\pi'|UG,q))-f(U(\pi|UG,q)))\\
    =&\sum_{UG\in\mathbbm{UG}}\rho_{UG}^q(f(\mE_{i\sim\mathcal{I}_{\mathcal{UG}^q}, \sigma\sim\pi'(\cdot|q)}[U(\sigma|i,q)])-f(\mE_{i\sim\mathcal{I}_{\mathcal{UG}^q}, \sigma\sim\pi(\cdot|q)}[U(\sigma|i,q)]))\\
    > & 0
\end{split}
\end{displaymath}
\end{proof}

which contradicts the assumption that $\pi$ maximizes the user fairness.

\begin{figure}[!tbp]
  \centering
  \includegraphics[width=0.95\textwidth]{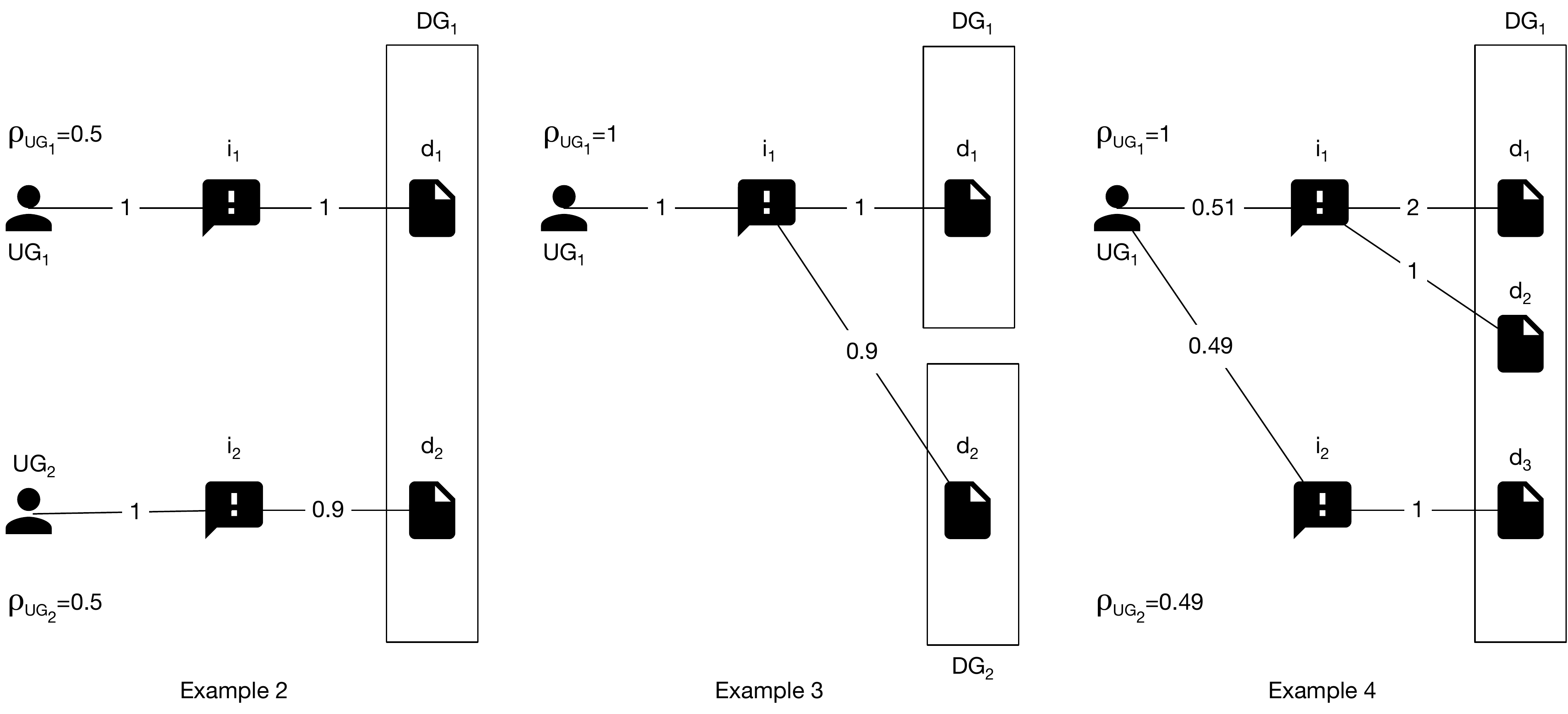}
  \caption{Examples to prove the theorems. Numbers on the edges between user groups and intents represent the intent distribution. The numbers of the edges between intents and items represent relevance. For clarity, we omit edges with $0$ probability or relevance.}
  \label{fig:example_supplementary}
\end{figure}

\begin{proof}[Proof of Theorem~\ref{theo:uf_efficiency_item}]
\label{proof:uf_efficiency_item}
Proof by construction. We first construct a ranking problem shown in \autoref{fig:example_supplementary} example 2. $\mathbbm{UG}=\{UG_1, UG_2\}$, $\rho_{UG_1} = \rho_{UG_2} =  0.5$, $\mathbbm{I} = \{i_1, i_2\}$, $\mathcal{I}_{\mathcal{UG}_1}(i_1) = \mathcal{I}_{\mathcal{UG}_2}(i_2) =1$ and $\mathcal{I}_{\mathcal{UG}_1}(i_2) = \mathcal{I}_{\mathcal{UG}_2}(i_1) =0$. 
$\mathbbm{DG} = \{DG_1\}$, $\mathbbm{D}^q = \{d_1, d_2\}$ and $d_1, d_2\in DG_1$. $r(d_1,i_1) = 1$, $r(d_2, i_2) = 0.9$ $r(d1, i_2) = r(d_2, i_1) = 0$. 

Without loss of generality, we assume $e(1) > e(2)$. Since there is only one item group, the policy that ensures items are ranked by their expected relevance within each item group is a deterministic ranking policy that ranks $d_1$ over $d_2$. 

We take a look at the user fairness of a ranking policy $\pi$, 
\begin{displaymath}
\begin{split}
    &UF(\pi|q)\\
    =& \sum_{UG\in\mathbbm{UG}}\rho_{UG}^qf(U(\pi|UG,q))\\
    =& \sum_{UG\in\mathbbm{UG}}\rho_{UG}^qf(\mE_{\sigma\sim\pi(\cdot|q),i\sim\mathcal{I}_{\mathcal{UG}}}\bigg[\sum_{d\in\mathbbm{D}^q}e(\sigma(d))r(d,i)\bigg])\\
    =&0.5f(\mE_{\sigma\sim\pi(\cdot|q)}\bigg[e(\sigma(d_1))\bigg]) + 0.5f(\mE_{\sigma\sim\pi(\cdot|q)}\bigg[0.9e(\sigma(d_2))\bigg])
\end{split}
\end{displaymath}

There are two possible rankings for this ranking problem and we denote the probability of the ranking that ranks $d_1$ overs over $d_2$ in the ranking policy as $x$. As long as the deterministic ranking policy is not optimal for user fairness under the exposure function $e$ and the the user fairness function $f$, the ranking policy that maximizes the user fairness does not ensure items are ranked by their expected relevance within each item group. For example, if we use the $log$ function as $f$, then the user fairness is \begin{displaymath}
\begin{split}
&UF(\pi|q)\\
=&0.5 log (e(1)x + e(2)(1-x)) + 0.5 log (0.9e(1)(1-x) + 0.9e(2)x)\\
=&0.5 log ((e(1)x + e(2)(1-x))(0.9e(1)(1-x) + 0.9e(2)x))\\
=&0.5log(0.9( -(e(1)-e(2))^2x^2 + (e(1)-e(2))^2x + e(1)e(2)))\\
\leq& 0.5log(0.225 (e(1)-e(2))^2), 
\end{split}
\end{displaymath}

and the equality holds when $x = 0.5$. So the ranking policy that maximizes user fairness selects the two rankings with equal probability. As a result, the items are not ranked by their expected relevance to the whole user population within each item group under the policy maximizing user fairness. 
\end{proof}

\subsection{Maximizing utility subject to disparate treatment constraints}

For the sake of brevity, we use ``maximize item fairness'' to refer to the more appropriately descriptive ``maximize utility subject to the disparate treatment constraints''.

\subsubsection{Zero-utility analysis of maximizing item fairness}

\begin{theorem}{(Maximizing item fairness can ensure non-zero utility for all item groups.)}
\label{theo:if_zero_item}
For any non-degenerate ranking problem, there always exists a merit function $M$ such that if a ranking policy $\pi$ maximizes utility subject to the disparate treatment constraints, then every item group has non-zero utility under $\pi$.
\end{theorem}
\begin{proof}
\label{proof:if_zero_item}
We construct a merit function $M$ such that every item group has the same merit. First, there always exists a ranking policy that satisfies the constraints since the policy that ranks items randomly satisfy the constraints. Second, any ranking policy that assigns zero exposure to an item group does not satisfy the constraints. Thus the ranking policy that maximizes utility subject to the constraints have non-zero utility for all item groups. 
\end{proof}

\begin{theorem}{(Maximizing item fairness can lead to zero utility for a user group.)}
\label{theo:if_zero_user}
There exists a non-degenerate ranking problem such that for any merit function $M$, the policy $\pi$ maximizing utility subject to the disparate treatment constraints has utility $U(\pi|UG,q)=0$ for a user group $UG$.
\end{theorem}

\begin{proof}
\label{proof:if_zero_user}
Proof by construction. We use the example in \autoref{fig:example} that was detailed in the Proof of Theorem~\ref{theo:util_zero_user}. We will show that $UG_1$ has zero utility under the policy maximizing item fairness no matter what merit function we choose.

For any merit function, the policy maximizing item fairness will always rank items in $d_{3*}$ over items in $d_{1*}$, $d_{2*}$ and rank items $d_{6*}$ over items in $d_{4*}$, $d_{5*}$. Thus items in $d_{3*}$ and $d_{6*}$ occupy all the exposure. And all the other items get 0 exposure. Since items that have non-zero relevance to $i_1$ and $i_2$ get zero exposure, $UG_1$ gets 0 utility.

\end{proof}

\begin{theorem}{(Maximizing item fairness can not ensure each ranking sampled from the policy covers the maximum amount of intent covered by any ranking.)}
\label{theo:if_zero_max_intent}
There exists a non-degenerate ranking problem such that for any merit function $M$, the amount of intent covered by any ranking sampled from the policy maximizing utility subject to the disparate treatment constraints is less than the maximum amount of intent covered by any ranking. 
\end{theorem}

\begin{proof}
\label{proof:if_zero_max_intent}
Proof by construction. We use the example in \autoref{fig:example} that was detailed in the Proof of Theorem~\ref{theo:util_zero_user}. The maximum amount of intent covered by any ranking is $1$ since we can select 3 items relevant to different intents and put them in the 3 non-zero exposure positions. We will show that $i_1$ and $i_2$ have zero utility under the policy maximizing item fairness no matter what merit function we choose. 

As shown in the Proof of Theorem~\ref{theo:if_zero_user}, for any merit function, items have non-zero relevance to $i_1$ and $i_2$ get zero exposure under the policy maximizing item fairness. As a result $i_1$ and $i_2$ have zero utility. So any ranking sampled from the policy maximizing item fairness has zero utility for $i_1$ and $i_2$. Thus any ranking sampled from the policy maximizing item fairness covers less intent than $1$, the maximum amount of intent covered by any ranking. 
\end{proof}

\subsubsection{Efficiency analysis of maximizing item fairness}

\begin{theorem}{(Maximizing item fairness ensures items are ranked by their expected relevance within each item group)}
\label{theo:if_pareto_items}
For any non-degenerate ranking problem and any satisfiable merit function $M$, if a ranking policy $\pi$ maximizes utility subject to the disparate treatment constraints, then
items within each item group are ranked by their expected relevance within each item group under $\pi$.
\end{theorem}
\begin{proof}
For any ranking policy that maximizes item fairness and any two items within the same group, they must be ranked by their expected relevance to the whole user population. If not, we can switch the two items without affecting the exposure of any item group --- and thus still satisfying the constraints ---, and achieve larger utility, which contradicts that the policy is a solution to maximizing utility subject to the constraints. 
\end{proof}

\begin{theorem}{(Maximizing item fairness is not Pareto efficient for the user groups)}
\label{theo:if_pareto_user}
There exists a non-degenerate ranking problem and a merit function $M$ such that any policy $\pi$ that maximizes utility subject to the disparate treatment constraints is not Pareto efficient for the user groups. 
\end{theorem}
\begin{proof}
\label{proof:if_pareto_user}
Consider the example 3 in \autoref{fig:example_supplementary}. The only policy that is Pareto efficient for the user groups is the one that maximizes utility for user group $G_1$ which is a deterministic ranking policy that ranks $d_1$ over $d_2$. We can construct a merit function $M(DG_1) = 1$, $M(DG_2) = 0.9$ and exposure function $e(1) = 1$ $e(2) = 0.5$ such that the disparate treatment constraints can only be satisfied when the ranking policy also ranks $d_2$ over $d_1$ sometimes. Thus the policy that maximizes utility subject to the constraints is not Pareto efficient for user groups. 
\end{proof}

\begin{theorem}{(Maximizing item fairness does not ensure each ranking sampled from the policy is Pareto efficient for the intents)}
\label{theo:if_pareto_intent}
There exists a non-degenerate ranking problem and a merit function $M$ such that for any policy $\pi$ that maximizes utility subject to the disparate treatment constraints, there exists a ranking $\sigma$ with $\pi(\sigma|q)>0$ such that $\sigma$ is not Pareto efficient for the intents. 
\end{theorem}
\begin{proof}
Still consider the example 3 in \autoref{fig:example_supplementary}. The only ranking that is Pareto efficient for the intents ranks $d_1$ over $d_2$. However, as in the Proof of Theorem~\ref{theo:if_pareto_user}, there is a merit function and an exposure function such that the ranking policy must rank $d_2$ over $d_1$ sometimes to satisfy the disparate treatment constraints. Thus there always exists a ranking in the policy maximizing item fairness that is not Pareto efficient for the intents. 
\end{proof}

\subsection{Maximizing Diversity}

\subsubsection{Zero-utility analysis of maximizing diversity}

\begin{theorem}{(Maximizing diversity can ensure that each ranking sampled from the policy maximizing diversity covers the maximum amount of intent covered by any ranking)}
\label{theo:D_zero_max_intent}
For any non-degenerate ranking problem, there always exists a diversity function $g$ such that if a ranking policy $\pi$ maximizes diversity, then any ranking sampled from $\pi$ covers the maximum amount of intent covered by any ranking.
\end{theorem}
\begin{proof}
\label{proof:D_zero_max_intent}
Any ranking policy that maximizes diversity consists of rankings that maximize diversity. If not we can replace the rankings that do not maximize diversity with the ones that maximize diversity to get larger diversity. Denote $\sigma$ one ranking that covers the maximum amount of intent covered by any ranking. We will construct a diversity function $g$ such that any ranking $\sigma'$ that covers smaller amount of intent than $\sigma$ has smaller diversity than $\sigma$. Denote the set of intents covered by $\sigma$ and $\sigma'$ as $\mathbbm{I}^{C}_{\sigma}$ and $\mathbbm{I}^{C}_{\sigma'}$. Denote the maximum utility of an intent over all intents and all rankings as $U_{max}^{\mathbbm{I}}=argmax_{i\in\mathbbm{I},\sigma''}U(\sigma''|i,q)$ and the minimum utility of an intent across all intents covered by $\sigma$ as $U_{min,\sigma}^{\mathbbm{I}}=argmin_{i\in\mathbbm{I},U(\sigma|i,q)>0}U(\sigma|i,q)$. The increasing concave function $g$ is constructed as a piece-wise linear function

\[ g(\cdot)=\begin{cases} 
      k_3(\cdot - t_2)  \quad\mbox{if}\quad\cdot \leq t_2  \\
      k_4(\cdot - t_2)  \quad\mbox{if}\quad\cdot > t_2
  \end{cases}
\]
with $k_3$ $k_4$ $t_2$ to be set later. We want to show that $D(\sigma|q)>D(\sigma'|q)$. 

\begin{equation}
\begin{split}
&D(\sigma|q) - D(\sigma'|q)\\
=&\mE_{i\sim\mathcal{I}_{\mathcal{U}^q}}[g(U(\sigma|i,q))] - \mE_{i\sim\mathcal{I}_{\mathcal{U}^q}}[g(U(\sigma'|i,q))]\\
=&\sum_{i\in\mathbbm{I}_\sigma^C}\mathcal{I}_{\mathcal{U}^q}(i)[g(U(\sigma|i,q))] + g(0)(1-\sum_{i\in\mathbbm{I}_\sigma^C}\mathcal{I}_{\mathcal{U}^q}(i)) - \sum_{i\in\mathbbm{I}_{\sigma'}^C}\mathcal{I}_{\mathcal{U}^q}(i)[g(U(\sigma'|i,q))] - g(0)(1-\sum_{i\in\mathbbm{I}_{\sigma'}^C}\mathcal{I}_{\mathcal{U}^q}(i))\\
= & \sum_{i\in\mathbbm{I}_\sigma^C}\mathcal{I}_{\mathcal{U}^q}(i)[g(U(\sigma|i,q))-g(0)] - \sum_{i\in\mathbbm{I}_{\sigma'}^C}\mathcal{I}_{\mathcal{U}^q}(i)[g(U(\sigma'|i,q)) - g(0)]\\
\geq & \sum_{i\in\mathbbm{I}_\sigma^C}\mathcal{I}_{\mathcal{U}^q}(i)[g(  U_{min,\sigma}^{\mathbbm{I}})-g(0)] - \sum_{i\in\mathbbm{I}_{\sigma'}^C}\mathcal{I}_{\mathcal{U}^q}(i)[g(U_{max}^{\mathbbm{I}} ) - g(0)]\\
&= k_3 U_{min,\sigma}^{\mathbbm{I}}\sum_{i\in\mathbbm{I}_\sigma^C}\mathcal{I}_{\mathcal{U}^q}(i) - k_3 U_{min,\sigma}^{\mathbbm{I}}\sum_{i\in\mathbbm{I}_{\sigma'}^C}\mathcal{I}_{\mathcal{U}^q}(i) - k_4 (U_{max}^{\mathbbm{I}} - U_{min,\sigma}^{\mathbbm{I}})\sum_{i\in\mathbbm{I}_{\sigma'}^C}\mathcal{I}_{\mathcal{U}^q}(i)  \quad (\mbox{assume }   t_2 = U_{min,\sigma}^{\mathbbm{I}} and k_4 > 0 and k_3 > 0)
\end{split}
\end{equation}

We set $t_2 =U_{min,\sigma}^{\mathbbm{I}}$, $k_4$ to be any positive real number, $k_3 >max(k_4,\frac{  (U_{max}^{\mathbbm{I}} - U_{min,\sigma}^{\mathbbm{I}})\sum_{i\in\mathbbm{I}_{\sigma'}^C}\mathcal{I}_{\mathcal{U}^q}(i) }{ U_{min,\sigma}^{\mathbbm{I}}(\sum_{i\in\mathbbm{I}_\sigma^C}\mathcal{I}_{\mathcal{U}^q}(i) - \sum_{i\in\mathbbm{I}_{\sigma'}^C}\mathcal{I}_{\mathcal{U}^q}(i) ) }k_4) $ so that $D(\sigma|q) - D(\sigma'|q)>0$. 
Since any ranking that covers smaller amount of intent gets less diversity than $\pi$, any ranking sampled from any ranking policy that maximizes diversity will cover the maximum amount of intent covered by any ranking. 
\end{proof}

\begin{theorem}{(Maximizing diversity can lead to $0$ utility for a user group.)}
\label{theo:D_zero_user}
There exists a non-degenerate ranking problem such that for any diversity function $g$, the policy $\pi$ maximizing diversity has utility $U(\pi|UG,q)=0$ for a user group $UG$.   
\end{theorem}
\begin{proof}
\label{proof:D_zero_user}
Consider the example in \autoref{fig:example_supplementary} example $2$. No matter what increasing concave diversity function $g$ we use, the policy maximizing diversity will be a deterministic policy that ranks $d_1$ over $d_2$. If the first position has non-zero exposure and the second position has zero exposure, the utility for $UG_2$ is always $0$. 
\end{proof}

\begin{theorem}{(Maximizing diversity can lead to zero utility for an item group)}
There exists a non-degenerate ranking problem such that for any diversity function $g$, the policy $\pi$ maximizing diversity has utility $U(\pi|DG,q)=0$ for an item group $DG$.
\end{theorem}

\begin{proof}
\label{proof:D_zero_item}
Proof by construction. We use the example in \autoref{fig:example} that was detailed in the Proof of Theorem~\ref{theo:util_zero_user}. We will show that $DG_2$ has zero utility under the policy maximizing diversity no matter what increasing concave diversity function $g$ we choose. 

For any increasing concave diversity function $g$, the policy maximizing diversity will always rank items in $d_{1*}$ over items in $d_{4*}$, rank items in $d_{2*}$ over items in  $d_{5*}$, and rank items in $d_{3*}$ over items in $d_{6*}$, since items in $d_{1*}$ $d_{2*}$ and $d_{3*}$ are the most relevant for intents $i_1$, $i_2$ and $i_3$. Thus items in $d_{1*}$, $d_{2*}$, $d_{3*}$ occupy all the exposure and all the other items get 0 exposure. Since items in $DG_2$ get 0 exposure, they have zero utility.

\end{proof}

\subsubsection{Efficiency analysis of maximizing diversity}

\begin{theorem}{(Each ranking sampled from the policy maximizing diversity is Pareto efficient for the intents)}
\label{theo:D_pareto_intent}
For any non-degenerate ranking problem and any diversity function $g$, if a ranking policy $\pi$ maximizes diversity, then any ranking $\sigma$ with $\pi(\sigma|q)>0$ is Pareto efficient for the intents.
\end{theorem}

\begin{proof}
Any ranking policy that maximizes diversity consists of rankings that maximize diversity. We only need to prove that if a ranking maximizes diversity, then it is Pareto efficient for the intents among all the rankings. 

Proof by contradiction. Assume there exists a ranking $\sigma$ such that $\sigma$ maximizes the diversity i.e. $\forall \sigma'$ $D(\sigma'|q)\leq D(\sigma|q)$ and $\sigma$ is not Pareto efficient for the intents. By the definition of Pareto efficiency, there exists a ranking $\sigma''$ such that $\sigma''$ dominates $\sigma$ for the intents. By the definition of dominance and the increasing property of $g$,
\begin{displaymath}
    D(\sigma''|q) - D(\sigma|q) =\mE_{i\sim\mathcal{I}_{\mathcal{U}^q}} \bigg[ g(U(\sigma|i,q)) -g(U(\sigma|i,q))  \bigg] > 0
\end{displaymath}
which contradicts the assumption that $\sigma$ maximizes the diversity. 
\end{proof}

\begin{theorem}{(Maximizing diversity is not Pareto efficient for the user groups)}
\label{theo:D_efficiency_user}
There exists a non-degenerate ranking problem and a diversity function $g$, such that any ranking policy $\pi$ that maximizes diversity is not Pareto efficient for the user groups. 
\end{theorem}
\begin{proof}
\label{proof:D_efficiency_user}
Consider the example 4 in \autoref{fig:example_supplementary}. The only ranking policy that is Pareto efficient for the user groups produces the ranking $(d_1, d_2, d_3)$ deterministically if we set $e(1) = 1$, $e(2) = 1$, $e(3) = 0$. We can always find a diversity function $g$ such that $d_3$ ranks at position $2$ or $1$ in the policy maximizing diversity according to Theorem~\ref{theo:D_zero_max_intent}. Under this diversity function $g$, the policy maximizing diversity is not Pareto efficient for the user groups. 
\end{proof}

\begin{theorem}{(Items are not ranked by their expected relevance within each item group in the policy maximizing the diversity. )}
\label{theo:D_efficiency_item}
There exists a non-degenerate ranking problem and a diversity function $g$, such that items are not ranked by their expected relevance within each item group in any policy that maximizes the diversity. 
\end{theorem}

\begin{proof}
\label{proof:D_efficiency_item}
Consider the example 4 in \autoref{fig:example_supplementary}. If we set $e(1) = 1$, $e(2) = 1$, $e(3) = 0$. We can always find a diversity fairness function $g$ such that $d_3$ ranks at position $2$ or $1$ in the policy maximizing diversity according to Theorem~\ref{theo:D_zero_max_intent}. Under this diversity function $g$, the items are not ranked by their expected relevance within each item group in any ranking sampled from any policy that maximizes diversity. 
\end{proof}

\subsection{\algnameNS}
\begin{theorem}{(\algname can ensure non-zero utility for all the user groups and item groups)}
For any non-degenerate ranking problem, there always exists a user fairness function $f$ and a merit function $M$ such that every user group and item group has non-zero utility under the ranking policy $\pi$ produced by \algnameNS. 
\end{theorem}
\begin{proof}

First, We construct a merit function $M$ such that every item group has the same merit. With this merit function, the first step of \algname will ensure each item group gets some exposure. In each item group, there exists at least one item that has non-zero relevance to an intent, otherwise the item group has zero optimal utility. As a result, in each item group, at least one item with non-zero expected relevance to the whole user population gets some exposure, otherwise we can switch the item with zero expected relevance that gets some exposure with an item with non-zero expected relevance that will increase user fairness no matter what user fairness function is. Thus each item group gets some utility if we use this merit function. 

Second, there always exists a ranking policy $\pi$ such that it has non-zero utility for every user group and satisfy the disparate treatment constraints with merit function $M$. The reason is that since every user group has non-zero optimal utility, there is at least one item has non-zero relevance to the intent the user has positive probability. We can construct a ranking policy such that this item gets some exposure and also satisfy the item fairness constraints. Then the convex sum of ranking policies that have non-zero utility for each user group is a ranking policy that has non-zero utility for every user group. 

As shown in the Proof of Theorem~\ref{theo:uf_zero_user}, for any ranking policy $\pi$ that has non-zero utility for every user group, we can construct a user fairness function $f$ such that no ranking policy with a user group having zero utility has larger user fairness than $\pi$.

As a result, under the constructed $f$ and the disparate treatment constraints with the constructed $M$, \algname always ensures non-zero utility for every user group and item group. 
\end{proof}

\begin{theorem}{(\algname can not ensure each ranking sampled from the policy covers the maximum amount of intent covered by any ranking. )}
There exists a non-degenerate ranking problem such that for any user fairness function $f$, merit function $M$ and diversity function $g$, the amount of intent covered by any ranking sampled from the policy produced by \algname is less than the maximum amount of intent covered by any ranking.  
\end{theorem}
\begin{proof}
Consider the example 4 in \autoref{fig:example_supplementary}. We set $e(1)=1, e(2) = 1, e(3) = 0$. No matter what user fairness function, merit function and diversity function we choose, \algname always produces the deterministic ranking $(d_1, d_2, d_3)$. As a result $i_2$ has zero utility. So the amount of intent covered by any ranking produced by \algname is less than the maximum amount of intent covered by any ranking. 
\end{proof}

\end{document}